\documentclass[a4paper, 11pt]{article}

\usepackage[utf8]{inputenc}
\usepackage{amsfonts, amsmath, amssymb, amsthm}
\usepackage{tikz}
\usetikzlibrary{arrows,shapes,snakes,automata, calc, chains, matrix, positioning, scopes, fit, backgrounds}
\usepackage{geometry}
\usepackage{xcolor}
\usepackage{xspace}
\usepackage{authblk}
\usepackage{multicol}
\usepackage{subcaption}

\newtheorem{theorem}{Theorem}
\newtheorem{corollary}{Corollary}

\newtheorem{lemma}{Lemma}

\newcommand{\ie}{\textit{i.e.\ }}

\newcommand{\Z}{\mathbb{Z}}

\newcommand{\priority}{\pi} 
 
\newcommand{\weight}{c}

\newcommand{\seq}[1]{\langle #1 \rangle}

\begin{document}

\title{A pseudo-quasi-polynomial algorithm \\ for mean-payoff parity games}

\author[1]{Laure Daviaud}
\author[1]{Marcin Jurdzi\'nski}
\author[1]{Ranko Lazi\'c}
\affil[1]{DIMAP, Department of Computer Science, University of Warwick, UK}
\date{}                     

\maketitle

\begin{abstract}
  In a mean-payoff parity game, one of the two players aims both to 
  achieve a qualitative parity objective and to minimize a quantitative
  long-term average of payoffs (aka.\ mean payoff).
  The game is zero-sum and hence the aim of the other player is to
  either foil the parity objective or to maximize the mean payoff.  

  Our main technical result is a pseudo-quasi-polynomial algorithm for
  solving mean-payoff parity games.
  All algorithms for the problem that have been developed for over a
  decade have a pseudo-polynomial and an exponential factors in their
  running times;
  in the running time of our algorithm the latter is replaced with a
  quasi-polynomial one.
  By the results of Chatterjee and Doyen (2012) and of Schewe,
  Weinert, and Zimmermann (2018), our main technical result implies
  that there are pseudo-quasi-polynomial algorithms for solving parity
  energy games and for solving parity games with weights. 

  Our main conceptual contributions are the definitions of strategy
  decompositions for both players, and a notion of
  progress measures for mean-payoff parity games that generalizes both
  parity and energy progress measures.
  The former provides normal forms for and succinct representations of
  winning strategies, and the latter enables the application to
  mean-payoff parity games of the order-theoretic machinery that
  underpins a recent quasi-polynomial algorithm for solving parity 
  games.    
\end{abstract}

\maketitle

%
%

\section{Introduction}
\label{section:introduction}
A motivation to study zero-sum two-player games on graphs comes 
from automata theory and logic, where they have been
used as a robust theoretical tool, for example, for streamlining 
of the initially notoriously complex proofs of Rabin's theorems on
the complementation of automata on infinite trees and the 
decidability of the monadic second-order logic on infinite 
trees~\cite{GH82,Zie98}, and for the development of the 
related theory of logics with fixpoint operators~\cite{EJ91}.
More practical motivations come from model checking and 
automated controller synthesis, where they serve as a clean 
combinatorial model for the study of the computational complexity
and algorithmic techniques for model checking~\cite{EJS01}, and 
for the automated synthesis of correct-by-design 
controllers~\cite{Tho95}. 
There is a rich literature on closely related ``dynamic games'' 
in the classical game theory and AI literatures reaching back to 
1950's, and games on graphs are also relevant to complexity 
theory~\cite{Con92} and to competitive ratio analysis of online 
algorithms~\cite{ZP96}. 

\subsection{Mean-payoff parity games}

A \emph{mean-payoff parity} game is played by two players---Con and
Dis---on a directed graph.
From the starting vertex, the players keep following edges of the
graph forever, thus forming an infinite path.
The set of vertices is partitioned into those owned by Con and those
owned by Dis, and it is the owner of the current vertex who picks
which outgoing edge to follow to the next current vertex.
Who is declared the winner of an infinite path formed by such
interaction is determined by the labels of vertices and edges
encountered on the path.
Every vertex is labelled by a positive integer called
its \emph{priority} and every edge is labelled by an integer called
its \emph{cost}.
The former are used to define the \emph{parity} condition:
the highest priority that occurs infinitely many times is \emph{odd}; 
and the latter are used to define the
(zero-threshold) \emph{mean-payoff} condition:   
the (lim-sup) long-run average of the costs is negative. 
If both the parity and the mean-payoff conditions hold then Con is
declared the winner, and otherwise Dis is. 
In the following picture, if Dis owns the vertex in the middle 
then she wins the game (with a positional strategy): 
she can for example always go to the left whenever she is in 
the middle vertex and this way achieve the positive mean 
payoff~$1/2$. 
Conversely, if Con owns the middle vertex then he wins the game. 
He can choose to go infinitely often to the left and see 
priority~$1$---in order to fulfill the parity condition---and 
immediately after each visit to the left, to go to the right a 
sufficient number of times---so as to make the mean-payoff negative. 
Note that a winning strategy for Con is not positional.

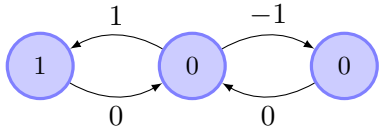
\begin{center}
\begin{tikzpicture}[scale=0.5]
\tikzset{
	every state/.style={draw=blue!50,very thick,fill=blue!20,scale=0.9},
	fleche/.style={->, >=latex}
}

\node[state] (a) at (-4,0) {$1$};
\node[state] (b) at (0,0) {$0$};
\node[state] (c) at (4,0) {$0$};

\path[fleche]     (b) edge [bend left] node [above] {$-1$} (c);
\path[fleche]     (c) edge [bend left] node [below] {$0$} (b);
\path[fleche]     (b) edge [bend right] node [above] {$1$} (a);
\path[fleche]     (a) edge [bend right] node [below] {$0$} (b);
\end{tikzpicture}
\end{center}

Throughout the paper, we write $V$ and $E$ for the sets of vertices
and directed edges in a mean-payoff parity game graph, $\pi(v)$ for
the priority of a vertex $v \in V$, and $c(v, u)$ for the cost of an 
edge $(v, u) \in E$.
Vertex priorities are positive integers no larger than~$d$, which we 
assume throughout the paper to be a positive even integer,
edge costs are integers whose absolute value does not exceed the 
positive integer~$C$,
and we write~$n$ and~$m$ for the numbers of vertices and edges in the
graph, respectively. 

Several variants of the algorithmic problem of
\emph{solving mean-payoff parity games} have been considered in the
literature.
The input always includes a game graph as described above.
The \emph{value} of (a vertex in) a mean-payoff parity game is defined 
as $\infty$ if Con does not have a winning strategy for the parity 
condition, and otherwise the smallest mean payoff that Con can secure
while playing so as to satisfy the parity condition. 
(Note that the paper that introduced mean-payoff parity games~\cite{CHJ05} 
defined Con to be the maximizer and not, as we do, the minimizer of the 
mean payoff. 
The two definitions are straightforwardly inter-reducible;
the choice we made allows for a better alignment of our key notion of
a mean-payoff parity progress measure with the literature on energy
progress measures~\cite{BCDGR11}.) 
The \emph{value problem} is to compute the value of every vertex.
The \emph{threshold problem} is, given an additional (rational) number
$\theta$ as a part of the input, to compute the set of vertices with
finite value (strictly) less than~$\theta$.
(Note that a value of a vertex is not finite, i.e., it is~$\infty$, if
and only if Con does not have a winning strategy for his parity
condition, which can be checked in quasi-polynomial
time~\cite{CJKLS17,JL17}.) 
In the \emph{zero-threshold problem} the threshold number~$\theta$ is
assumed to be~$0$. 

As Chatterjee et al.~\cite[Theorem 10]{CHS17} have shown, the
threshold problem can be used to solve the value problem at the cost
of increasing the running time by the modest $O(n \cdot \log(nC))$
multiplicative term.
Their result, together with a routine linear-time reduction from the
threshold problem to the zero-threshold problem (subtract $\theta$
from costs of all edges), motivate us to focus on solving the
zero-threshold problem in this paper. 
For brevity, we will henceforth write ``mean-payoff condition''
instead of ``zero-threshold mean-payoff condition''. 

The roles of the two players in a mean-payoff parity
game are not symmetric for several reasons. 
One is that Con aims to satisfy a \emph{conjunction} of the parity
condition and of the mean-payoff condition, while the
goal of Dis is to satisfy a \emph{disjunction} of the negated
conditions.  
The other one is that negations of the parity condition and of the 
mean-payoff condition are not literally the parity and the mean-payoff
conditions, respectively:  
the negation of the parity condition swaps the roles of even and odd, 
and the negation of the strict (``less than'') mean-payoff condition
is non-strict (``at least'').   
The former asymmetry (conjunction vs disjunction) is material and
our treatments of strategy construction for players Con and Dis differ
substantially, but the latter are technically benign.
The discussion above implies that the goal of player Dis is to either
satisfy the parity condition in which the highest priority that occurs
infinitely many times is even, or to satisfy the ``at least''
zero-threshold mean-payoff condition.  

\subsection{Related work}

Mean-payoff games have been studied since 1960's and there is a rich body
of work on them in the stochastic games literature.
We selectively mention the positional determinacy result of Ehrenfeucht and 
Mycielski~\cite{EM79} (i.e., that positional optimal strategies exist for
both players), and the work of Zwick and Paterson~\cite{ZP96}, who pointed out
that positional determinacy implies that deciding the winner in mean-payoff 
games is both in NP and in co-NP, and gave a pseudo-polynomial algorithm 
for computing values in mean-payoff games that runs in time $O(mn^3C)$.
Brim et al.~\cite{BCDGR11} introduced energy progress measures as natural
witnesses for winning strategies in closely related energy games, they 
developed a lifting algorithm to compute the least energy progress measures,
and they observed that this leads to an algorithm for computing values in 
mean-payoff games whose running time is $O(mn^2C \cdot \log(nC))$, which is 
better than the algorithm of Zwick and Paterson~\cite{ZP96} if $C = 2^{o(n)}$. 
Comin and Rizzi~\cite{CR17} have further refined the usage of the lifting 
algorithm for energy games achieving running time $O(mn^2C)$. 

Parity games have been studied in the theory of automata on infinite 
trees, fixpoint logics, and in verification and synthesis since early 
1990's~\cite{EJ91,EJS01}. 
Very selectively, we mention early and influential recursive algorithms by 
McNaughton~\cite{McN93} and by Zielonka~\cite{Zie98}, the running times 
of which are $O(n^{d+O(1)})$. 
The breakthrough result of Calude et al.~\cite{CJKLS17} gave the first 
algorithm that achieved an $n^{o(d)}$ running time. 
Its running time is polynomial $O(n^5)$ if $d \leq \log n$ and 
quasipolynomial $O(n^{\lg d + 6})$ in general. 
(Throughout the paper, we write $\lg x$ to denote $\log_2 x$, 
and we write $\log x$ when the base of the logarithm is moot.) 
Note that Calude et al.'s polynomial bound for $d \leq \log n$ implies 
that parity games are FPT (fixed parameter tractable) when the number~$d$ 
of distinct vertex priorities is the parameter. 
Further analysis by Jurdzi\'nski and Lazi\'c~\cite{JL17} established that 
running times $O(mn^{2.38})$ for $d \leq \lg n$, and 
$O(dmn^{\lg(d/\lg n)+1.45})$ for $d = \omega(\lg n)$, can be achieved using 
their succinct progress measures, and Fearnley et al.~\cite{FJSSW17} obtained 
similar results by refining the technique and the analysis of Calude et 
al.~\cite{CJKLS17}. 
Existence of polynomial-time algorithms for solving parity games and 
for solving mean-payoff games are fundamental long-standing open 
problems~\cite{EJS01,ZP96,Joh07}. 

Mean-payoff parity games have been introduced by Chatterjee et al.~\cite{CHJ05}
as a proof of concept in developing algorithmic techniques for solving 
games (and hence for controller synthesis) which combine qualitative 
(functional) and quantitative (performance) objectives. 
Their algorithm for the value problem is inspired by the recursive algorithms 
of McNaughton~\cite{McN93} and Zielonka~\cite{Zie98} for parity games, from 
which its running time acquires the exponential dependence $mn^{d+O(1)}C$ on 
the number of vertex priorities.  
Chatterjee and Doyen~\cite{CD12} have simplified the approach by considering 
energy parity games first, achieving running time $O(dmn^{d+4}C)$ for the 
threshold problem, which was further improved by Bouyer et al.~\cite{BMOU11} 
to $O(mn^{d+2}C)$ for the  value problem. 
Finally, Chatterjee et al.~\cite{CHS17} have achieved the running time
$O(mn^{d}C \log(nC))$ for the value problem, but their key original 
technical results are for the two special cases of mean-payoff parity games 
that allow only two distinct vertex priorities, for which they achieve 
running time~$O(mnC)$ for the threshold problem, by using amortized analysis 
techniques from dynamic algorithms. 
Note that none of those algorithms escapes the exponential dependence on the 
number of distinct vertex priorities, simply because they all follow the 
recursive structure of the algorithms by McNaughton~\cite{McN93} and by 
Zielonka~\cite{Zie98}. 

Other quantitative extensions of parity games have been considered;
for example, Fijalkow and Zimmermann~\cite{FZ14} introduced 
\emph{parity games with costs}, and Schewe, Weinert, and
Zimmermann~\cite{SWZ18} generalized those to 
\emph{parity games with weights}.  
Chatterjee and Doyen~\cite{CD12} have proved that the problem of
deciding the winner in mean-payoff parity games is log-space
equivalent to the problem of deciding the winner in energy parity
games, and Schewe et al.~\cite{SWZ18} have proved that the latter is 
polynomial-time equivalent to the problem of deciding the winner in
parity games with weights. 
It follows that the three problems, of deciding the winner in mean-payoff parity games, in energy parity games, and in parity games
with weights, respectively, are polynomial-time equivalent.

\subsection{Our contributions}

Our main technical result is the first pseudo-quasi-polynomial
algorithm for solving mean-payoff parity games.
More specifi\-cally, we prove that the threshold problem 
can be solved in pseudo-polynomial time
$mn^{2+o(1)}C$ for $d = o(\log n)$, in pseudo-polynomial time 
$mn^{O(1)}C$ if $d = O(\log n)$ 
(where the constant in the exponent of~$n$ depends logarithmically on the 
constant hidden in the big-Oh expression $O(\log n)$), 
and in pseudo-quasi-polynomial time $O(dmn^{\lg(d/\lg n)+2.45}C)$ 
if $d = \omega(\log n)$. 
By~\cite[Theorem 10]{CHS17}, we obtain running times for solving the value 
problem that are obtained from the ones above by multiplying them by the 
$O(n \log(nC))$ term. 

Our key conceptual contributions are the notions of strategy decompositions 
for both players in mean-payoff parity games, and of mean-payoff parity
progress measures.
The former explicitly reveal the underlying strategy structure of winning 
sets for both players, and they provide normal forms and succinct 
representations for winning strategies. 
The latter provide an alternative form of a witness and a normal form 
of winning strategies for player Dis, which make explicit the 
order-theoretic structures that underpin the original progress measure 
lifting algorithms for parity~\cite{Jur00} and energy games~\cite{BCDGR11}, 
respectively, as well as the recent quasi-polynomial succinct progress 
measure lifting algorithm for parity games~\cite{JL17}. 
The proofs of existence of strategy decompositions follow the well-beaten
track of using McNaughton-Zielonka-like inductive arguments, and existence 
of progress measures that witness winning strategies for Dis is established 
by extracting them from strategy decompositions for Dis.

Our notion of mean-payoff parity progress measures combines features of 
parity and energy progress measures, respectively.
Crucially, our mean-payoff progress measures inherit the ordered tree 
structure from parity progress measures, and the additional numerical 
labels of vertices (that capture the energy progress measure aspects) do 
not interfere substantially with it.  
This allows us to directly apply the combinatorial ordered tree coding 
result by Jurdzi\'nski and Lazi\'c~\cite{JL17}, which limits the 
search space in which the witnesses are sought by the lifting procedure 
to a pseudo-quasi-polynomial size, yielding our main result.
The order-theoretic properties that the lifting procedure relies on 
naturally imply the existence of the least 
(in an appropriate order-theoretic sense) progress measure, from which 
a positional winning strategy for Dis on her winning set can be easily 
extracted.

In order to synthesize a strategy decomposition---and hence a winning
strategy---for Con in pseudo-quasi-polynomial time, we take a different 
approach. 
Progress measures for games typically yield positional winning strategies
for the relevant player~\cite{KK95,Jur00,BCDGR11}, but optimal strategies
for Con in mean-payoff parity games may require finite memory (or even infinite memory in the variant where Con has to ensure a non-positive mean-payoff~\cite{CHJ05}).  
That motivates us to forgo attempting to pin a notion of progress measures 
to witness winning strategies for Con. 
We argue, instead, that a McNaughton-Zielonka-style recursive procedure 
can be modified to run in pseudo-quasi-polynomial time and produce 
a strategy decomposition of Con's winning set.
The key insight is to avoid invoking some of the recursive calls, and 
instead to replace them by invocations of the pseudo-quasi-polynomial 
lifting procedure for Dis, merely to compute the winning set for 
Dis---and hence also for Con, because by determinacy Con has a winning 
strategy whenever Dis does not. 
As a result, each invocation of the recursive procedure only makes 
recursive calls on disjoint subgames, which makes it perform only a 
polynomial number of steps other than invocations of the lifting 
procedure, overall yielding a pseudo-quasi-polynomial algorithm.

Note that our pseudo-quasi-polynomial algorithm for mean-payoff parity
games can be used to solve energy
parity games and parity games with weights in pseudo-quasi-polynomial time. Indeed, deciding the
winner in the latter two classes of games is polynomial-time
equivalent to deciding the winner in mean-payoff games by the results
of Chatterjee and Doyen~\cite{CD12} and Schewe et al.~\cite{SWZ18},
respectively.

\paragraph*{Organisation of the paper.}
In Section~\ref{section:strategy-decompositions}, we define 
strategy decompositions for Dis and Con, and we prove that they 
exist if and only if the respective player has a winning strategy.
In Section~\ref{section:progress-measures}, we define progress 
measures for Dis, and we prove that such a progress measure exists 
if and only if Dis has a strategy decomposition. 
In Section~\ref{section:lifting}, we give a 
pseudo-quasi-polynomial lifting algorithm for computing the least 
progress measure, from which a strategy decomposition for Dis 
of her winning set, and the winning set for Con, can be derived.
In Section~\ref{section:winsets-to-decompositions}, we show how
to also compute a strategy decomposition for Con on his winning set 
in pseudo-quasi-polynomial time, using the lifting procedure to speed 
up a NcNaughton-Zielonka-style recursive procedure.

\section{Strategy decompositions}
\label{section:strategy-decompositions}

In this section we introduce our first key concept of 
\emph{strategy decompositions} for each of the two players.
They are hierarchically defined objects, of size polynomial in the
number of vertices in the game graph, that witness existence of 
winning strategies for each of the two players on their winning 
sets.
Such decompositions are implicit in earlier literature, in particular
in algorithms for mean-payoff parity games~\cite{CHJ05,CD12,BMOU11,CHS17} 
that follow the recursive logic of McNaughton's~\cite{McN93} and 
Zielonka's~\cite{Zie98} algorithms for parity games. 
We make them explicit because we belive that it provides conceptual
clarity and technical advantages. 
Strategy decompositions pinpoint the recursive strategic structure of the
winning sets in mean-payoff parity games 
(and, by specialization, in parity games too), which may provide 
valuable insights for future work on the subject. 
What they allow us to do in this work is to streamline the proof that
the other key concept we introduce---mean-payoff parity progress 
measures---witness existence of winning strategies for Dis.

We define the notions of strategy decompositions for Dis and for Con, 
then in Lemmas~\ref{lemma:strategy-for-Dis} and~\ref{lemma:strategy-for-Con}
we prove that the decompositions naturally yield winning strategies for
the corresponding players, and finally in 
Lemma~\ref{lemma:existence-strategy-decomposition} we establish that 
in every mean-payoff game, both players have strategy decompositions of
their winning sets. 
The proofs of all three lemmas mostly use well-known inductive 
McNaughton-Zielonka-type arguments that should be familiar to anyone 
who is conversant in the existing literature on mean-payoff parity games.
We wish to think that for a curious non-expert, this section offers 
a streamlined and self-contained exposition of the key algorithmic
ideas behind earlier works on mean-payoff parity 
games~\cite{CHJ05,CD12,BMOU11}.

\subsection{Preliminaries}

Notions of strategies, positional strategies, plays, plays consistent
with a strategy, winning strategies, winning sets, reachability 
strategies, traps, mean payoff, etc., are defined in the usual way. 
We forgo tediously repeating the definitions of those common and 
routine concepts, referring a non-expert but interested reader to 
consult the (typically one-page) Preliminaries or Definitions 
sections of any of the previously published papers on mean-payoff 
parity games~\cite{CHJ05,CD12,BMOU11,CHS17}.
One notable difference between our set-up and those found in the
above-mentioned papers is that for an infinite sequence of numbers 
$\seq{c_1, c_2, c_3, \dots}$, we define its mean payoff 
to be $\limsup_{n \to \infty} (1/n) \cdot \sum_{i=1}^n c_i$, rather 
than the more common  
$\liminf_{n \to \infty} (1/n) \cdot \sum_{i=1}^n c_i$;
this is because we chose Con to be the minimizer of the mean payoff, 
instead of the typical choice of making him the maximizer.

\subsection{Strategy decompositions for Dis}
\label{subsection:strategy-decomposition-for-or}

Let $W \subseteq V$ be a subgame (\ie~a non empty induced subgraph of $V$ with no deadend) in which the biggest vertex priority 
is~$b$. 
We define strategy decompositions for Dis by induction on~$b$ and the
size of~$W$.
We say that $\omega$ is a $b$-decomposition of
$W$ for Dis if the following conditions hold (pictured in Figure~\ref{figure:decompositionDis}).
\begin{enumerate}
\item
  If $b$ is even then
  $\omega = \big((R, \omega'), (T, \tau), B\big)$,  
  such that:
  \begin{enumerate}
  \item
    sets $R$, $T$, and~$B \not= \emptyset$ are a partition of~$W$; 
  \item
    $B$ is the set of vertices of the top priority~$b$
    in~$W$;   
  \item
    $\tau$ is a positional reachability strategy for Dis from~$T$
    to~$B$ in~$W$;  
  \item
    $\omega'$ is a $b'$-decomposition of~$R$ for~Dis, where $b' < b$. 
  \end{enumerate}
\item
  If $b$ is odd then
  $\omega = \big((U, \omega''), (T, \tau), (R, \omega')\big)$, such
  that:
  \begin{enumerate}
  \item
    sets $U$, $T$, and~$R \not= \emptyset$ are a partition of~$W$; 
  \item
    \label{item:even-recurrentDis}
    $\omega'$ is either: 
    \begin{enumerate}
    \item
    \label{item:win-by-inductionDis}
      a $b'$-decomposition of~$R$ for Dis, where $b' < b$; or
    \item
      \label{item:win-by-energyDis}
      a positional strategy for Dis that is mean-payoff winning for
      her on~$R$; 
    \end{enumerate}
  \item
    $\tau$ is a positional reachability strategy for Dis
    from~$T$ to~$R$ in~$W$;
  \item
    $\omega''$ is a $b''$-decomposition of~$U$ for Dis, where
    $b'' \leq b$;   
  \item
    $R$ is a trap for Con in~$W$. 
  \end{enumerate}
\end{enumerate}
We say that a subgame $W$ has a strategy decomposition for Dis if it has a
$b$-decomposition for some~$b$.
A heuristic, if somewhat non-standard, way to think about 
sets~$T$ and~$R$ in the above definition is that sets denoted 
by~$T$ are \emph{transient} and sets denoted by~$R$ are
\emph{recurrent}. 
The meanings of those words here are different than in, say, 
Markov chains, and refer to strategic, rather than 
probabilistic, properties.

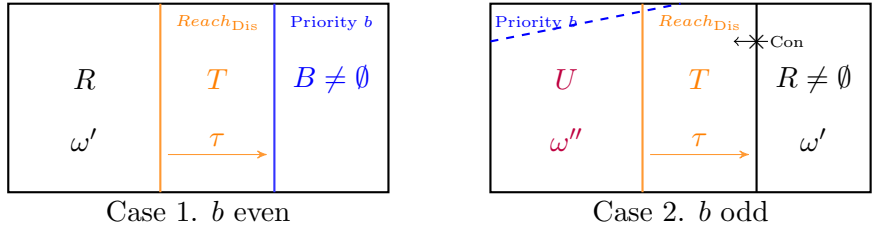
\begin{figure}
\begin{center}
\begin{tikzpicture}
\draw[thick] (0,0) rectangle (5,2.5);
\draw[blue!80, thick] (3.5,0) -- (3.5,2.5);
\draw[orange!80, thick] (2,0) -- (2,2.5);
\node[] at (4.25,2.25) {\textcolor{blue}{\tiny Priority $b$}};
\node[] at (4.25,1.5) {\textcolor{blue}{$B \neq \emptyset$}};
\node[] at (2.75,2.25) {\textcolor{orange}{\tiny $Reach_\text{Dis}$}};
\node[] at (2.75,1.5) {\textcolor{orange}{$T$}};
\node[] at (2.75,0.7) {\textcolor{orange}{$\tau$}};
\draw[->, >=stealth', orange!80] (2.1,0.5) -- (3.4,0.5);
\node[] at (1,1.5) {$R$};
\node[] at (1,0.7) {$\omega'$};
\node[] at (2.5,-0.25) {Case 1. $b$ even};
\end{tikzpicture} \hspace{1cm}
\begin{tikzpicture}
\draw[thick] (0,0) rectangle (5,2.5);
\draw[thick] (3.5,0) -- (3.5,2.5);
\draw[orange!80, thick] (2,0) -- (2,2.5);
\node[] at (4.25,1.5) {$R \neq \emptyset$};
\node[] at (4.25,0.7) {$\omega'$};
\node[] at (2.75,2.25) {\textcolor{orange}{\tiny $Reach_\text{Dis}$}};
\node[] at (2.75,1.5) {\textcolor{orange}{$T$}};
\node[] at (2.75,0.7) {\textcolor{orange}{$\tau$}};
\draw[->, >=stealth', orange!80] (2.1,0.5) -- (3.4,0.5);
\node[] at (1,1.5) {\textcolor{purple}{$U$}};
\node[] at (1,0.7) {\textcolor{purple}{$\omega''$}};
\draw[dashed, blue, thick] (0,2) -- (2.5,2.5);
\node[] at (0.6,2.25) {\textcolor{blue}{\tiny Priority $b$}};
\node[] at (3.9,2)  {\tiny Con};
\draw[->] (3.65,2) -- (3.2,2);
\node[] at (3.5,2) {$\times$};
\node[] at (2.5,-0.25) {Case 2. $b$ odd};
\end{tikzpicture}
\end{center}
\caption{\label{figure:decompositionDis}Strategy decompositions for Dis.}
\end{figure}

Given a strategy decomposition~$\omega$ for Dis, we inductively define a
positional strategy $\sigma(\omega)$ for Dis in the following way: 
\[
\sigma(\omega) =
\begin{cases}
  \sigma(\omega') \cup \tau \cup \beta &
    \text{if } \omega = \big((R, \omega'), (T, \tau), B\big), \\
  \sigma(\omega'') \cup \tau \cup \sigma(\omega') &
    \text{if } \omega =
    \big((U, \omega''), (T, \tau), (R, \omega')\big),     
\end{cases}
\]
where $\beta$ is an arbitrary positional strategy for Dis on~$B$,
and $\sigma(\omega') = \omega'$ in case~\ref{item:win-by-energyDis}. 

\begin{lemma}
  \label{lemma:strategy-for-Dis}
  If $\omega$ is a strategy decomposition of~$W$ for Dis
  and $W$ is a trap for Con, then $\sigma(\omega)$ is a
  positional winning strategy for Dis from every vertex in~$W$. 
\end{lemma}

\begin{proof}
  We proceed by induction on the number of vertices in $W$. 
  The reasoning involved in the base cases 
  (when $R = \emptyset$ or $U = \emptyset$) 
  is analogous and simpler than in the inductive cases, 
  hence we immediately proceed to the latter. 

  We consider two cases based on the parity of the biggest vertex 
  priority~$b$ in~$W$. 
  
  First, assume that $b$ is even and let 
  $\omega = \big((R, \omega'), (T, \tau), B\big)$ be a $b$-decomposition
  of~$W$.
  We argue that every infinite play consistent with~$\sigma(\omega)$ is
  winning for Dis.
  If it visits vertices in~$B$ infinitely many times then the parity
  condition for Dis is satisfied because~$b$ is the biggest vertex 
  priority and it is even.
  Otherwise, it must be the case that the play visits vertices 
  in~$T \cup B$ only finitely many times, because visiting a vertex
  in~$T$ always leads in finitely many steps to visiting a vertex 
  in~$B$ by following the reachability strategy~$\tau$. 
  Therefore, eventually the play never leaves~$R$ and is consistent with
  strategy~$\sigma(\omega')$, which is winning for Dis by the inductive
  hypothesis.

  Next, assume that $b$ is odd.
  Let $\omega = \big((U, \omega''), (T, \tau), (R, \omega')\big)$
  be a $b$-decomposition.
  We argue that every infinite play consistent with~$\sigma(\omega)$ is
  winning for Dis.
  If it visits $T \cup R$, then by following strategy~$\tau$, it 
  eventually reaches and never leaves~$R$ 
  (because $R$ is a trap for Con), and hence it is winning for
  Dis because $\sigma(\omega')$ is a winning strategy for Dis by the
  inductive hypothesis, or by condition~\ref{item:win-by-energyDis}. 
  Otherwise, if such a play never visits $T \cup R$ then it is winning
  for Dis because $\sigma(\omega'')$ is a winning strategy for Dis by
  the inductive hypothesis.
\end{proof}

\subsection{Strategy decompositions for~Con}
\label{subsection:strategy-decomposition-for-square}

Let $W \subseteq V$ be a subgame in which the biggest vertex priority
is~$b$. 
We define strategy decompositions for Con by induction on~$b$ and the
size of~$W$.
We say that $\omega$ is a $b$-decomposition of~$W$
for Con if the following conditions hold (pictured in Figure~\ref{figure:decompositionCon}).
\begin{enumerate}
\item
  If $b$ is odd then
  $\omega = \big((R, \omega'), (T, \tau), B, \lambda\big)$,   
  such that:
  \begin{enumerate}
  \item
    sets $R$, $T$, and~$B \not= \emptyset$ are a partition of~$W$; 
  \item
    $B$ is the set of vertices of priority~$b$ in~$W$; 
  \item
    $\tau$ is a positional reachability strategy for Con
    from~$T$ to~$B$ in~$W$; 
  \item
    $\omega'$ is a $b'$-decomposition of~$R$ for Con,
    where $b' < b$;
    \item
  \label{item:energy-strategyCon}
  $\lambda$ is a positional strategy for Con that is mean-payoff
  winning for him on~$W$. 
  \end{enumerate}
\item
  If $b$ is even then
  $\omega = \big((U, \omega''), (T, \tau), (R, \omega')\big)$, 
  such that:
  \begin{enumerate}
  \item
    sets $U$, $T$, and~$R \not= \emptyset$ are a partition of~$W$; 
  \item
    \label{item:even-recurrentCon}
    $\omega'$ is a $b'$-decomposition of~$R$ for Con,
    where $b' < b$; 
  \item
    $\tau$ is a positional reachability strategy for Con from~$T$
    to~$R$ in~$W$; 
  \item
    $\omega''$ is a $b''$-decomposition of~$U$ for Con,
    where $b'' \leq b$; 
  \item
    $R$ is a trap for Dis in~$W$.
    \end{enumerate}  
\end{enumerate}
We say that a subgame has a strategy decomposition for Con if
it has a $b$-decomposition for some~$b$.
Note that the definition is analogous to that of a strategy
decomposition for Dis in most aspects, with the following
differences:  
\begin{itemize}
\item
  the roles of Dis and Con, and of even and odd, are swapped; 
\item
  the condition~\ref{item:even-recurrentDis} is simplified;
\item
  an extra component~$\lambda$, and the 
  condition~\ref{item:energy-strategyCon}, are added.
\end{itemize}

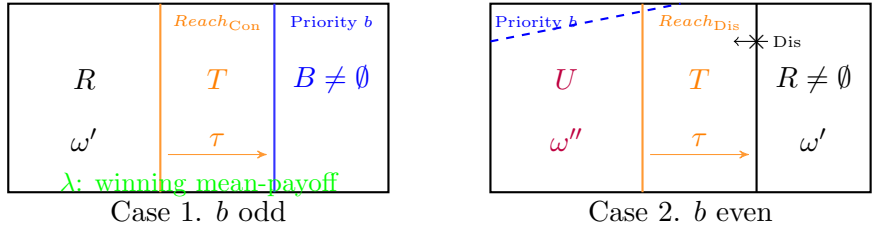
\begin{figure}
\begin{center}
\begin{tikzpicture}
\draw[thick] (0,0) rectangle (5,2.5);
\draw[blue!80, thick] (3.5,0) -- (3.5,2.5);
\draw[orange!80, thick] (2,0) -- (2,2.5);
\node[] at (4.25,2.25) {\textcolor{blue}{\tiny Priority $b$}};
\node[] at (4.25,1.5) {\textcolor{blue}{$B \neq \emptyset$}};
\node[] at (2.75,2.25) {\textcolor{orange}{\tiny $Reach_\text{Con}$}};
\node[] at (2.75,1.5) {\textcolor{orange}{$T$}};
\node[] at (2.75,0.7) {\textcolor{orange}{$\tau$}};
\draw[->, >=stealth', orange!80] (2.1,0.5) -- (3.4,0.5);
\node[] at (1,1.5) {$R$};
\node[] at (1,0.7) {$\omega'$};
\node[] at (2.5,0.1) {\small \textcolor{green}{$\lambda$: winning mean-payoff}};
\node[] at (2.5,-0.25) {Case 1. $b$ odd};
\end{tikzpicture} \hspace{1cm}
\begin{tikzpicture}
\draw[thick] (0,0) rectangle (5,2.5);
\draw[thick] (3.5,0) -- (3.5,2.5);
\draw[orange!80, thick] (2,0) -- (2,2.5);
\node[] at (4.25,1.5) {$R \neq \emptyset$};
\node[] at (4.25,0.7) {$\omega'$};
\node[] at (2.75,2.25) {\textcolor{orange}{\tiny $Reach_\text{Dis}$}};
\node[] at (2.75,1.5) {\textcolor{orange}{$T$}};
\node[] at (2.75,0.7) {\textcolor{orange}{$\tau$}};
\draw[->, >=stealth', orange!80] (2.1,0.5) -- (3.4,0.5);
\node[] at (1,1.5) {\textcolor{purple}{$U$}};
\node[] at (1,0.7) {\textcolor{purple}{$\omega''$}};
\draw[dashed, blue, thick] (0,2) -- (2.5,2.5);
\node[] at (0.6,2.25) {\textcolor{blue}{\tiny Priority $b$}};
\node[] at (3.9,2)  {\tiny Dis};
\draw[->] (3.65,2) -- (3.2,2);
\node[] at (3.5,2) {$\times$};
\node[] at (2.5,-0.25) {Case 2. $b$ even};
\end{tikzpicture}
\end{center}
\caption{\label{figure:decompositionCon}Strategy decompositions for Con.}
\end{figure}

Given a strategy decomposition~$\omega$ for Con, 
we inductively define a strategy $\sigma(\omega)$ 
for Con in the following way: 
\begin{itemize}
\item 
  If $b$ is odd and $\omega = ((R,\omega'),(T,\tau),B,\lambda)$, 
  then the strategy proceeds in (possibly infinitely many) rounds.
  Round~$i$, for $i = 1, 2, 3, \dots$, involves the following steps:
  \begin{enumerate}
  \item 
    if starting in $R$, follow $\sigma(\omega')$ for as long as staying in $R$;
  \item 
    if starting in $T$, or having arrived there from~$R$, 
    follow~$\tau$ until $B$ is reached;
  \item once $B$ is reached, follow $\lambda$ for $n + (2n + 3^n + 2)nC$ steps and 
    proceed to round~$i+1$. 
  \end{enumerate}
\item If $b$ is even and $\omega = ((U,\omega''),(T,\tau),(R,\omega'))$, then let:
  \[
  \sigma(\omega) = \sigma(\omega'') \cup \tau \cup \sigma(\omega'). 
  \]
\end{itemize}

\begin{lemma}
  \label{lemma:strategy-for-Con}
  If $\omega$ is a strategy decomposition of~$W$ for Con 
  and $W$ is a trap for Dis, then $\sigma(\omega)$ is a 
  winning strategy for Con from every vertex in~$W$.
\end{lemma}

\begin{proof}
  We proceed by induction on the number of vertices in $W$,
  omitting the base cases 
  (when $R = \emptyset$, or $U = \emptyset$, respectively), 
  since they are analogous and simpler than the inductive cases.
  We strengthen the inductive hypothesis by requiring that: 
  If $\omega$ is a strategy decomposition of~$W$ for Con 
  and $W$ is a trap for Dis, then $\sigma(\omega)$ is a 
  winning strategy for Con from every vertex in~$W$ and the sum of the costs 
  of the edges in any finite 
  play consistent with $\sigma(\omega)$ is bounded by $(n_W+3^{n_W})C$, where $n_W$ is the number 
  of vertices in $W$ (and recall that $C$ is the maximal cost on all the edges).

  We consider two cases based on the parity of~$b$.
  First, assume that $b$ is even and let 
  $\omega = \big((U, \omega''), (T, \tau), (R, \omega')\big)$.
  Observe that a play consistent 
  with~$\sigma(\omega) = \sigma(\omega'') \cup \tau \cup \sigma(\omega')$ 
  either never leaves~$U$, or if it does then after a finite number of 
  steps (following the reachability strategy~$\tau$) it enters~$R$ 
  and then never leaves it because~$R$ is a trap for Dis.
  It follows that the play is winning for Con by the inductive 
  hypothesis, because it is eventually either consistent
  with strategy~$\sigma(\omega'')$ or with~$\sigma(\omega')$. Moreover, 
  let us write $n_U$ (resp. $n_T$, $n_R$) for the number of vertices in $U$ (resp. $T$,
  $R$). Every finite play consistent with such a strategy can be decomposed into
  a play consistent with $\sigma(\omega'')$, a play going from $U$ to $T$, 
  consistent with $\tau$ and reaching $R$ 
  (thus using at most $n_T + 1$ edges)
  and a play consistent with $\sigma(\omega')$ (any of these plays can be empty).
  Suppose that none of those plays is empty (the other cases can be handled similarly). 
  In particular, $n_U$ and $n_R$ are smaller than $n_W$.
  By inductive hypothesis, the sum of the costs 
  of the edges in any of such finite plays is bounded by 
  $(n_U+3^{n_U})C + (n_T + 1)C + (n_R+3^{n_R})C$, and:
  $$(n_U+3^{n_U})C + (n_T + 1)C + (n_R+3^{n_R})C \leq (n_U+ n_T + n_R + 3.3^{n_W-1})C \leq (n_W+3^{n_W})C$$

  Next, assume that $b$ is odd, and let 
  $\omega = \big((R, \omega'), (T, \tau), B, \lambda\big)$
  be a $b$-decomposition.

  
Let us first prove that any infinite play consistent with $\sigma(\omega)$ 
is winning for Con. If after a finite number of steps, 
the play reaches $R$ and never leaves it, then $\sigma(\omega)$ is compatible with $\sigma(\omega')$ which is winning for Con by induction hypothesis (because $B$ is non-empty). Otherwise, vertices in $B$ or $T$ are seen infinitely often. In that case, vertices in $B$ are seen infinitely often (by contradiction, if not, then necessarily we go through point 2. and 3. in the strategy definition a finite number of times, and so after a finite number of steps, the play is in point 1. forever). Because $B$ is the set of vertices of highest priority $b$ which is odd, Con wins the parity game.
Let us prove that the play has also negative mean-payoff. 
Let us write $n_R$ (resp. $n_T$, $n_B$) for the number of vertices in $R$ (resp. $T$,
  $B$).
By hypothesis the play can be decomposed into (infinitely many) finite plays, each of them decomposed into three consecutive (possibly empty) plays $p_1$, $p_2$ and $p_3$ as follows:
\begin{itemize}
\item $p_1$ consists of vertices in $R$ and is consistent with $\sigma(\omega')$ (point $1$.), 
\item $p_2$ goes from $R$ to $B$, consists of vertices in $T$ and is consistent with the reachability strategy to reach $B$. Then it contains at most $n_T+1$ edges and is of cost at most $(n_T+1)C$ (point 2.),
\item $p_3$ is consistent with $\lambda$ and uses $n_W + (2n_W + 3^{n_W} + 2)n_WC$ edges. A negative cycle is thus 
necessarily reached and the sum of the costs of the edges of $p_3$ is at most $n_WC - (2n_W + 3^{n_W} + 2)C$ (point 3.).
\end{itemize}
It is sufficient to prove that such a finite play $p = p_1 p_2 p_3$ has negative mean-payoff. 
By inductive hypothesis,
the sum of the costs of the edges of such a play is at most 
$(n_R+3^{n_R})C + (n_T+1)C + n_WC - (2n_W + 3^{n_W} + 2)C$ which is negative.

It remains to prove that along a finite play the sum of the costs of the edges 
never exceeds $(n_W+3^{n_W})C$, which is true using the decomposition above, the inductive
hypothesis and the fact that the sum of the costs of the edges
on a play consistent with $\lambda$ will never exceed~$n_WC$. Thus,
the maximum cost of such a finite play is $(n_R + 3^{n_R} + n_T + 1 + n_W)C$
which is smaller than $(n_R + 3^{n_W-1} + n_T + 1 + 3^{n_W-1})C$, or again 
$(n_W+3^{n_W})C$.

%
\end{proof}

\subsection{Existence of strategy decompositions}

In the following lemma, we prove that every game can be partitioned
into two sets of vertices, so that there is a strategy decomposition
of one for Dis, and a strategy decomposition of the other one for Con.  
Those sets correspond to the winning sets for Dis and Con,
respectively. 

\begin{lemma}
\label{lemma:existence-strategy-decomposition}
  There is a partition $W_{\mathrm{Dis}}$ and $W_{\mathrm{Con}}$ of~$V$, such that there is a
  strategy decomposition of~$W_{\mathrm{Dis}}$ for Dis 
  (provided $W_{\mathrm{Dis}} \neq \emptyset$) 
  and a strategy decomposition
  of~$W_{\mathrm{Con}}$ for Con (provided $W_{\mathrm{Con}} \neq \emptyset$).  
\end{lemma}

The proof of Lemma~\ref{lemma:existence-strategy-decomposition} follows 
the usual template of using a McNaughton-Zielonka inductive argument,
as adapted to mean-payoff parity games by Chatterjee et al.~\cite{CHJ05}, and 
then simplified for threshold mean-payoff parity games by 
Chatterjee et al.~\cite[Appendix~C]{CHS17}. 

\begin{proof}
  The proof is by induction on the size of the game graph. 
  We strengthen the induction hypothesis by also requiring 
  that $W_{\mathrm{Dis}}$ and $W_{\mathrm{Con}}$ are traps 
  in~$V$ for respectively Con and Dis. 
  The base case of one vertex is straightforward. 
  Let $b$ be the highest vertex priority, and let $B$ be the set of
  vertices of the highest priority~$b$.
  We consider two cases depending on the parity of~$b$. 
  
  \medskip
  
  The first case is when $b$ is even.
  Let $T$ be the set of vertices (not including vertices in~$B$) from
  which Dis has a strategy to reach a vertex in~$B$, and let $\tau$ be
  a corresponding positional reachability strategy. 

  Let $R = V \setminus (B \cup T)$. 
  By the inductive hypothesis, there is a partition $W'_{\mathrm{Dis}}$ and $W'_{\mathrm{Con}}$
  of~$R$, such that there is a strategy decomposition~$\omega'_{\mathrm{Dis}}$
  of~$W'_{\mathrm{Dis}}$ for Dis, and there is a strategy
  decomposition~$\omega'_{\mathrm{Con}}$ of~$W'_{\mathrm{Con}}$ for Con.
  If $W'_{\mathrm{Con}} = \emptyset$ then $\omega'_{\mathrm{Dis}}$ is a $b'$-decomposition of~$R$
  for Dis, where $b' < b$, and hence
  $\big((R, \omega'_{\mathrm{Dis}}), (T, \tau), B\big)$ is a $b$-decomposition
  of~$V$ for Dis. So $W_{\mathrm{Dis}} = V$ and $W_{\mathrm{Con}} = \emptyset$ 
  fulfils the conditions of the lemma.

  If $W'_{\mathrm{Con}} \not= \emptyset$, then let $T'$ be the set of vertices
  (not including vertices in~$W'_{\mathrm{Con}}$) from which Con has a strategy to
  reach a vertex in~$W'_{\mathrm{Con}}$, and let $\tau'$ be a corresponding
  positional reachability strategy.
  Let $U = V \setminus (W'_{\mathrm{Con}} \cup T')$. 
  By the inductive hypothesis, there is a partition $W''_{\mathrm{Dis}}$ and
  $W''_{\mathrm{Con}}$ of~$U$, such that there is a strategy
  decomposition~$\omega''_{\mathrm{Dis}}$ of~$W''_{\mathrm{Dis}}$ for Dis, and a strategy
  decomposition~$\omega''_{\mathrm{Con}}$ of~$W''_{\mathrm{Con}}$ for Con. Moreover, $W''_{\mathrm{Dis}}$ and
  $W''_{\mathrm{Con}}$ are traps for respectively Con and Dis in~$U$.

  We claim that $W''_{\mathrm{Dis}}$ and $W''_{\mathrm{Con}} \cup T' \cup W'_{\mathrm{Con}}$ is a partition
  of~$V$, traps for respectively Con and Dis, such that there is a strategy decomposition of the former
  for Dis, and there is a strategy decomposition of the latter for
  Con.
  The former is straightforward: $W''_{\mathrm{Dis}}$ is a trap for Con in $U$ which is itself a trap for
  Con in $V$ by construction, so $W''_{\mathrm{Dis}}$ is a trap for Con in $V$. 
  Moreover, $\omega''_{\mathrm{Dis}}$ is a
  strategy decomposition of~$W''_D$ for Dis.
  For the latter, $W''_{\mathrm{Con}} \cup T' \cup W'_{\mathrm{Con}}$ is a trap for Dis by construction and 
  we claim that
  $\omega$ is a strategy decomposition
  of~$W''_{\mathrm{Con}} \cup T' \cup W'_{\mathrm{Con}}$ for Con, where 
  $\omega = \big((W''_{\mathrm{Con}}, \omega''_{\mathrm{Con}}), (T', \tau'), 
    (W'_{\mathrm{Con}}, \omega'_{\mathrm{Con}})\big)$. 
  Indeed, $W'_{\mathrm{Con}}$ is non-empty, is a trap for Dis by induction hypothesis and does not contain any vertices of priority $b$ by construction. Thus, $\omega'_{\mathrm{Con}}$ is a $b'$-decomposition of
$W'_{\mathrm{Con}}$ for Con where $b'<b$. Similarly, by induction hypothesis, $\omega''_{\mathrm{Con}}$ is a $b''$-decomposition of
$W'_{\mathrm{Con}}$ for Con where $b''\leq b$. 
  
   \medskip
  
  The second case is when $b$ is odd. 
  Let $R$ be the set of vertices winning for Dis for the mean-payoff game. 
  
  First, suppose that $R$ is non empty, and let 
  $U = V \setminus R$. By the inductive hypothesis, there is a partition $W'_{\mathrm{Dis}}$ and $W'_{\mathrm{Con}}$
  of~$U$, such that there is a strategy decomposition~$\omega'_{\mathrm{Dis}}$
  of~$W'_{\mathrm{Dis}}$ for Dis, and there is a strategy
  decomposition~$\omega'_{\mathrm{Con}}$ of~$W'_{\mathrm{Con}}$ for Con. Moreover, $W'_{\mathrm{Dis}}$ and $W'_{\mathrm{Con}}$
  are traps in $U$ for respectively Con and Dis.
  We claim that $W'_{\mathrm{Con}}$ and $W'_{\mathrm{Dis}} \cup R$ is a partition
  of~$V$, traps for respectively Dis and Con, such that there is a strategy decomposition of the former
  for Con, and there is a strategy decomposition of the latter for
  Dis.
  The former is straightforward: $W'_{\mathrm{Con}}$ is a trap for Dis in $U$ which is itself a trap for
  Dis in $V$ by construction (because $R$ is a winning set for Dis), so $W'_{\mathrm{Con}}$ is a trap for Dis in $V$. 
  Moreover, $\omega'_{\mathrm{Con}}$ is a
  strategy decomposition of~$W'_{\mathrm{Con}}$ for Con.
  For the latter, $W'_{\mathrm{Dis}} \cup R$ is a trap for Con by construction and 
  we claim that
  $\omega$ is a strategy decomposition
  of~$W'_{\mathrm{Dis}} \cup R$ for Dis, where 
  $\omega = \big((W'_{\mathrm{Dis}}, \omega'_{\mathrm{Dis}}), (\emptyset, \emptyset), 
    (R, \omega')\big)$, with $\omega'$ to be a positional strategy for Dis 
    that is mean-payoff winning for
      her on~$R$.
  Indeed, $R$ is non-empty, is a trap for Con by definition and $\omega'$ is a mean-payoff winning positional strategy for Dis on it. Moreover, by induction hypothesis, $\omega'_{\mathrm{Dis}}$ is a $b'$-decomposition of
$W'_{\mathrm{Dis}}$ for Con where $b'\leq b$. 
  
  Suppose now that $R$ is empty, that is to say that there exist $\lambda$,
  a positional strategy for Con 
    that is mean-payoff winning for him on $V$.
  Let $T$ be the set of vertices (not including vertices in~$B$) from
  which Con has a strategy to reach a vertex in~$B$, and let $\tau$ be
  a corresponding positional reachability strategy. 

  Let $R' = V \setminus (B \cup T)$. 
  By the inductive hypothesis, there is a partition $W'_{\mathrm{Dis}}$ and $W'_{\mathrm{Con}}$
  of~$R$, such that there is a strategy decomposition~$\omega'_{\mathrm{Dis}}$
  of~$W'_{\mathrm{Dis}}$ for Dis, and there is a strategy
  decomposition~$\omega'_{\mathrm{Con}}$ of~$W'_{\mathrm{Con}}$ for Con.

  If $W'_{\mathrm{Dis}} = \emptyset$ then $\omega'_{\mathrm{Con}}$ is a $b'$-decomposition of~$R$
  for Con, where $b' < b$ and thus 
  $\omega = \big((R', \omega'_{\mathrm{Con}}), (T, \tau), B, \lambda \big)$ is a strategy decomposition
  of $V$ for Con.
  
  Otherwise (if $W'_{\mathrm{Dis}} \neq \emptyset$), then let $T'$ be the set of vertices 
  (not including vertices in~$W'_{\mathrm{Dis}}$)
  from which Dis has a strategy to reach a vertex in~$W'_{\mathrm{Dis}}$, and let
  $\tau'$ be a corresponding positional reachability strategy.
  Let $U' = V \setminus (W'_{\mathrm{Dis}} \cup T')$.
  By the inductive hypothesis, there is a partition~$W''_{\mathrm{Dis}}$
  and~$W''_{\mathrm{Con}}$ of~$U'$, such that there is a strategy
  decomposition~$\omega''_{\mathrm{Dis}}$ of~$W''_{\mathrm{Dis}}$ for Dis, and a strategy
  decomposition~$\omega''_{\mathrm{Con}}$ of~$W''_{\mathrm{Con}}$ for Con.
  
  We claim that $W''_{\mathrm{Con}}$ and $W''_{\mathrm{Dis}} \cup T' \cup W'_{\mathrm{Dis}}$ is a partition
  of~$V$, traps for respectively Dis and Con, such that there is a strategy 
  decomposition of the former
  for Con, and there is a strategy decomposition of the latter for
  Dis.
    The former is straightforward: $W''_{\mathrm{Con}}$ is a trap for Dis in $U'$ which is itself a trap for
  Dis in $V$ by construction, so $W''_{\mathrm{Con}}$ is a trap for Dis in $V$. 
  Moreover, $\omega''_{\mathrm{Con}}$ is a
  strategy decomposition of~$W''_{\mathrm{Con}}$ for Con.
  For the latter, $W''_{\mathrm{Dis}} \cup T' \cup W'_{\mathrm{Dis}}$ is a trap for Con by construction and 
  we claim that
  $\omega$ is a strategy decomposition
  of~$W''_{\mathrm{Dis}} \cup T' \cup W'_{\mathrm{Dis}}$ for Dis, where 
  $\omega = \big((W''_{\mathrm{Dis}}, \omega''_{\mathrm{Dis}}), (T', \tau), 
    (W'_{\mathrm{Dis}}, \omega'_{\mathrm{Dis}})\big)$.
  Indeed, $W'_{\mathrm{Dis}}$ is non-empty, is a trap for Con by induction hypothesis and does not contain any vertices of priority $b$ by construction. Thus, $\omega'_{\mathrm{Dis}}$ is a $b'$-decomposition of
$W'_{\mathrm{Dis}}$ for Dis where $b'<b$. Similarly, by induction hypothesis, $\omega''_{\mathrm{Dis}}$ is a $b''$-decomposition of
$W''_{\mathrm{Dis}}$ for Dis where $b''\leq b$.
\end{proof}

Observe that Lemmas~\ref{lemma:strategy-for-Dis}, 
\ref{lemma:strategy-for-Con}, and~\ref{lemma:existence-strategy-decomposition} 
form a self-contained argument to establish both determinacy of threshold
mean-payoff parity games (from every vertex, one of the players has a winning 
strategy), and membership of the problem of deciding the winner both in NP 
and in co-NP. 
For the latter, it suffices to note that strategy decompositions can be 
described in a polynomial number of bits, and it can be routinely checked in 
small polynomial time whether a proposed strategy decomposition for either of 
the players satisfies all the conditions in the corresponding definition. 
The NP and co-NP membership has been first established by Chatterjee and 
Doyen~\cite{CD12}; 
we merely give an alternative proof.

\begin{corollary}[Chatterjee and Doyen~\cite{CD12}]
  The problem of deciding the winner in mean-payoff parity games is both in NP
  and in co-NP. 
\end{corollary}

\section{Mean-payoff parity progress measures}
\label{section:progress-measures}

In this section we introduce the other key 
concept---\emph{mean-payoff parity progress measures}---that plays 
the critical role in achieving our main technical result---the 
first pseudo-quasi-polynomial algorithm for solving mean-payoff 
parity games.
In Lemmas~\ref{lemma:measure-to-decomposition} 
and~\ref{lemma:decomposition-to-measure} we establish that 
mean-payoff parity progress measures witness existence of winning
strategies for Dis, by providing explicit translations between them 
and strategy decompositions for Dis.

We stress that the purpose of introducing yet another concept of
witnesses for winning strategies for Dis is to shift technical 
focus from highlighting the recursive strategic structure of 
winning sets in strategy decompositions, to an order theoretic
formalization that makes the recursive structure be reflected
in the concept of ordered trees.
The order-theoretic formalization then allows us---in 
Section~\ref{section:lifting}---to apply the combinatorial result
on succinct coding of ordered trees by Jurdzi\'nski and 
Lazi\'c~\cite{JL17}, paving the way to the pseudo-quasi-polynomial 
algorithm.

\subsection{The definition}

A progress measure maps every vertex with an element of a linearly ordered set. 
Edges along which those elements decrease (for another specific defined order) are
called progressive, and an infinite path consisting only of progressive edges is winning 
for Dis. Then, we can derive a winning strategy for Dis if she can always follow a progressive edge and if Con has no other choice than following a progressive edge. 

Recall the assumption that~$d$---the upper bound on the vertex
priorities---is even.  

A \emph{progress measurement} is a sequence
$\big(\seq{m_{d-1}, m_{d-3}, \dots, m_\ell}, e\big)$,
where:
\begin{itemize}
\item
  $\ell$ is odd and $1 \leq \ell \leq d+1$
  (note that if $\ell = d+1$ then
  $\seq{m_{d-1}, m_{d-3}, \dots, m_\ell}$ is the empty
  sequence~$\seq{}$); 
\item
  $m_i$ is an element of a linearly ordered set
  (for simplicity, we write~$\leq$ for the order relation), for each
  odd~$i$, such that $\ell \leq i \leq d-1$;
\item
  $e$ is an integer such that $0 \leq e \leq nC$, or $e = \infty$. 
\end{itemize}
A \emph{progress labelling} $(\mu, \varphi)$ maps vertices to
progress measurements in such a way that if vertex~$v$ is mapped to 
$$\big(\mu(v), \varphi(v)\big) =
\big(\seq{m_{d-1}, m_{d-3}, \dots, m_\ell}, e\big)$$
then
\begin{itemize}
\item
  $\ell \geq \pi(v)$; and
\item
  if $e = \infty$ then $\ell$ is the smallest odd integer such that
  $\ell \geq \pi(v)$. 
\end{itemize}

For every priority $p$, $1 \leq p \leq d$, we obtain
a \emph{$p$-truncation} $\seq{m_{d-1}, m_{d-3}, \dots, m_\ell}|_p$ of $\seq{m_{d-1}, m_{d-3}, \dots, m_\ell}$, by removing the
components corresponding to all odd priorities smaller than~$p$. 
For example, if we fix $d=8$ then we have
$\seq{a, b, c}|_8 = \seq{}$, $\seq{a, b, c}|_6 = \seq{a}$, and
$\seq{a, b, c}|_3 = \seq{a, b, c}|_2 = \seq{a, b, c}$. 
We compare sequences using the lexicographic order; for simplicity,
and overloading notation, we write~$\leq$ to denote it.
For example, $\seq{a} < \seq{a, b}$, and $\seq{a, b, c} < \seq{a, d}$
if $b < d$.  

Let $(\mu, \varphi)$ be a progress labelling.
Observe that---by de\-fi\-ni\-tion---$\mu(v)|_{\pi(v)} = \mu(v)$, for every
vertex~$v \in V$.
We say that an edge $(v, u) \in E$ is \emph{progressive} 
in~$(\mu, \varphi)$ if: 
\begin{enumerate}
\item
  $\mu(v) > \mu(u)|_{\pi(v)}$; or
\item
  $\mu(v) = \mu(u)|_{\pi(v)}$, $\pi(v)$ is even, and
  $\varphi(v) = \infty$; or
\item
  $\mu(v) = \mu(u)$, $\varphi(v) \not= \infty$, and 
  $\varphi(v) + \weight(v, u) \geq \varphi(u)$.
\end{enumerate}

We can represent tuples as nodes in a tree where the components of the tuple represent
the branching directions in the tree to go from the root to the node. For example, a tuple $\seq{a, b, c}$ corresponds to the node reached from the root by first reaching the $a$th child of the root, then the $b$th child of this latter and finally the $c$th child of this one.
This way, the notion of progressive edges can be seen on a tree as in Figure~\ref{figure:progressive}.

\begin{figure}
\begin{multicols}{2}
\begin{tikzpicture}[sibling distance=2.5em, level distance=3.5em]
\tikzstyle{one}=[scale=0.5, shape=circle, fill=black]
\node[one] (a) {}
    child { node[one] {} 
    	child { node[one] {} 
    		child { node {$\vdots$} }
    		child { node {} edge from parent[draw=none] } }
    	child { node[one] {} } }	
	child { node[one] {} }
    child { node[one] (b) {} 
    	child { node[one] {} }
    	child { node[one] (c) {} 
    		child { node[one] (h) {}
    			 child [grow=left, level distance=1.5em] {node {$\mu(v)$} edge from parent[draw=none] } 
    			 child { node[one] (i) {} 
    			 	 child { node (l) {$\vdots$} }
    			     child { node (m) {$\vdots$} } }
    			 child { node[one] (j) {} 
    			 	 child { node {} edge from parent[draw=none] } 
    			     child { node (n) {$\vdots$} } }    			  
    			 child { node[one] (k) {} 
    			 	 child { node {} edge from parent[draw=none] }
    			     child { node (o) {$\vdots$} } } } 
    		child { node[one] (d) {} 
    			child { node {} edge from parent[draw=none] } 
    			child { node {} edge from parent[draw=none] } 
    			child { node {} edge from parent[draw=none] }
    			child { node[one] (e) {} 
    				child { node {} edge from parent[draw=none] }
    				child { node {} edge from parent[draw=none] }
    				child { node (p) {$\vdots$} } } } }
    	child { node[one] (g) {}  
    		child { node {} edge from parent[draw=none] }
    		child { node {} edge from parent[draw=none] } 
    		child { node (q) {$\vdots$} } } }
    child { node[one] (f) {} 
        child { node {} edge from parent[draw=none] }
        child { node {} edge from parent[draw=none] } 
        child { node (r) {$\vdots$} } };
        
   \begin{pgfonlayer}{background}
        \node [fill=blue!20, fit=(a), rounded corners=4mm, rotate=0, minimum width=2em, minimum height=2em] {};
        \node [fill=blue!20, fit=(b), rounded corners=4mm, rotate=0, minimum width=2em, minimum height=2em] {};
        \node [fill=blue!20, fit=(c), rounded corners=4mm, rotate=0, minimum width=2em, minimum height=2em] {};
        \node [fill=blue!20, fit=(d), rounded corners=4mm, rotate=0, minimum width=2em, minimum height=2em] {};
        \node [fill=blue!20, fit=(e), rounded corners=4mm, rotate=0, minimum width=2em, minimum height=2em] {};
        \node [fill=blue!20, fit=(f), rounded corners=4mm, rotate=0, minimum width=2em, minimum height=2em] {};
        \node [fill=blue!20, fit=(g), rounded corners=4mm, rotate=0, minimum width=2em, minimum height=2em] {};
        \node [fill=blue!20, fit=(p), rounded corners=4mm, rotate=0, minimum width=2em, minimum height=2em] {};
        \node [fill=blue!20, fit=(q), rounded corners=4mm, rotate=0, minimum width=2em, minimum height=2em] {};
        \node [fill=blue!20, fit=(r), rounded corners=4mm, rotate=0, minimum width=2em, minimum height=2em] {};

        \node [fill=orange!20, fit=(h), rounded corners=4mm, rotate=0, minimum width=2em, minimum height=2em] {};
         \node [fill=green!20, fit=(h), rounded corners=4mm, rotate=0, minimum width=1.5em, minimum height=1.5em] {};
        \node [fill=orange!20, fit=(i), rounded corners=4mm, rotate=0, minimum width=2em, minimum height=2em] {};
        \node [fill=orange!20, fit=(j), rounded corners=4mm, rotate=0, minimum width=2em, minimum height=2em] {};
        \node [fill=orange!20, fit=(k), rounded corners=4mm, rotate=0, minimum width=2em, minimum height=2em] {};
        \node [fill=orange!20, fit=(l), rounded corners=4mm, rotate=0, minimum width=2em, minimum height=2em] {};
        \node [fill=orange!20, fit=(m), rounded corners=4mm, rotate=0, minimum width=2em, minimum height=2em] {};
        \node [fill=orange!20, fit=(n), rounded corners=4mm, rotate=0, minimum width=2em, minimum height=2em] {};
        \node [fill=orange!20, fit=(o), rounded corners=4mm, rotate=0, minimum width=2em, minimum height=2em] {};
    \end{pgfonlayer}
\end{tikzpicture}

\columnbreak

The siblings are ordered according to the linear order $\leq$. The smallest child is on the right and the greatest on the left in the picture.
An edge $(v,u)$ is progressive if one of the three following conditions holds:

\bigskip

\textcolor{blue!50}{- condition 1 -} \\
$\mu(u)$ is one of the blue nodes, \ie above or on the right of $\mu(v)$.

\bigskip

\textcolor{orange!50}{- condition 2 -} \\
$\priority(v)$ is even, $\varphi(v) = \infty$ and $\mu(u)$ is one of the orange nodes, \ie belongs to the subtree rooted in $\mu(v)$.

\bigskip

\textcolor{green!50}{- condition 3 -} \\
$\mu(u) = \mu(v)$, $\varphi(u) \in \Z$ and
$\varphi(v) + \weight(v,u) \geq \varphi(u)$.

\end{multicols}
\caption{\label{figure:progressive}Conditions for an edge to be progressive.}
\end{figure}

A progress labelling $(\mu, \varphi)$ is a \emph{progress measure}
if:
\begin{itemize}
\item
  for every vertex owned by Dis, there is at least one outgoing edge
  that is progressive in~$(\mu, \varphi)$; and
\item
  for every vertex owned by Con, all outgoing edges are progressive
  in~$(\mu, \varphi)$. 
\end{itemize}

In the next two sections, we prove that there is a strategy decomposition 
for Dis if and only is there is a progress measure.

\subsection{From progress measures to strategy decompositions}

\begin{lemma}
\label{lemma:measure-to-decomposition}
  If there is a progress measure then there is a strategy
  decomposition of~$V$ for Dis.  
\end{lemma}

In the proof we will use the following simple fact 
(see, for example, Brim et al.~\cite{BCDGR11}):
if all the edges in an infinite path are progressive and 
fulfill condition~$3$.\ of the definition, then the mean payoff 
of this path is non-negative (and thus winning for Dis).

\begin{proof}
  We proceed by induction on the number of distinct vertex priorities
  in the game graph. 
  Let $b\leq d$ be the highest priority appearing in the game.    

  The base case is when~$b$ is the only vertex priority.  
  If $b$ is even, then by setting $B = V$ and $T = R = \emptyset$ we
  obtain a strategy decomposition of~$V$ for Dis.
  If $b$ is odd, then an edge can only be progressive if it satisfies
  condition~3.\ of the definition of a progressive edge;
  hence the progress measure yields a positional strategy~$\omega'$
  for Dis that is mean-payoff winning for her on~$V$. 
  It follows that setting $R = V$, $T = U = \emptyset$, and $\omega'$
  as above, we obtain a strategy decomposition of~$V$ for Dis.
  
Consider the inductive step now. First, suppose that $b$ is even. 
Let $B$ be the set of the vertices of priority $b$. 
Let $T$ be the set of vertices from which Dis has a reachability strategy 
to $B$, $\tau$ be this positional strategy and let 
$R =  V \setminus (B \cup T)$. 
Because, by construction, there is no edge from a vertex in $R$ owned 
by Dis to a vertex in $B \cup T$, the progress measure on $V$ gives 
also a progress measure on $R$ when restricted to its vertices. 
Let $\omega$ be a $b'$-decomposition of~$R$ for Dis that exists by the 
inductive hypothesis.
Note that $b' < b$ because the biggest priority in~$R$ is smaller than~$b$.
It follows that $\big((R, \omega), (T, \tau), B\big)$ is a strategy 
decomposition for Dis in~$V$.
  
  Suppose now that $b$ is odd. 
  Let $R$ be the set of vertices labelled by the smallest tuple: 
  $R =  \{ v \in V : \mu(v) \leq \mu(u) \text{ for all } u \in V \}$.
  (If we pictured the tuples on a tree as in 
  Figure~\ref{figure:progressive}, those would be the vertices 
  that are mapped to the rightmost-top node in the tree among the 
  nodes at least one vertex is mapped to.)
  Let $R'$ be the subset of $R$ of those vertices having a 
  finite~$\varphi$: $\{ v \in R : \varphi(v) \neq \infty\}$.
  
  Suppose first that $R'\neq \emptyset$. An edge going out from a 
  vertex in~$R'$ can only be progressive if it fulfills condition~$3$.\ 
  in the definition. 
  It then has to go to a vertex of $R'$ too.
  Thus, $R'$ is a trap for Con, and Dis has a winning strategy $\omega'$ 
  in $R'$ for the mean-payoff game. 
  
  Let $T$ be the set of vertices from which Dis has a strategy to 
  reach $R'$ and let $\tau $ this positional reachability strategy. 
  Let $U = V \setminus (R' \cup T)$. 
  Because, by construction, there is no edge from a vertex in $U$ 
  owned by Dis to a vertex in $R' \cup T$, then the progress measure 
  on~$V$ gives also a progress measure on $U$ when restricted to its 
  vertices. 
  We can then apply the inductive hypothesis and get $\omega$ a 
  strategy decomposition of~$U$ for Dis. 
  Note that $\big((U, \omega), (T, \tau), (R',\omega')\big)$ is  
  a strategy decomposition of~$V$ for Dis.
  
  Suppose now that $R' = \emptyset$. The non-empty set $R$ contains 
  only vertices $v$ such that $\varphi(v) = \infty$. 
  Then, by definition and because all those vertices are associated 
  with the same tuple, they must all have priority $b'$ or $b'+1$ for 
  some even number~$b'$. 
  
  Any edge going out from a vertex of $R$ is progressive if and only 
  if it fulfills condition~$2$.\ of the definition. 
  Thus, the priority of all the vertices in $R$ has to be even and 
  is consequently $b'$ with $b' < b$. 
  
  Let $R'' =  \{ u \in V : \mu(v) = \mu(u)|_{\pi(v)} \text{ for } v \in R \}$.
  (If we picture the tuples on a tree as in Figure~\ref{figure:progressive}, 
  those are the vertices that are mapped to the nodes in the subtree 
  rooted in the node corresponding to~$R$.)
  By definition, the priority of all those vertices is also smaller 
  than~$b$. 
  Moreover, an edge going out from a vertex in $R''$ can only be 
  progressive if it goes to a vertex in $R''$ too. 
  So, $R''$ is a trap for Con and an edge from a vertex in $R''$ owned 
  by Dis to a vertex not in~$R''$ cannot be progressive. 
  So the progress measure on $V$ gives also a progress measure on $R''$ 
  when restricted to its vertices. 
  By the inductive hypothesis, there is a strategy decomposition 
  $\omega''$ of $R''$ for Dis.
  Let $T$ be the set of vertices from which Dis has a strategy to reach 
  $R''$ and let $\tau$ be a corresponding positional reachability strategy. 
  Let $U = V \setminus (R'' \cup T)$. 
  Because, by construction, there is no edge in~$U$ from a vertex
  owned by Dis to a vertex in $R'' \cup T$, the progress measure on $V$ 
  gives also a progress measure on~$U$ when restricted to its vertices. 
  By the inductive hypothesis, there is a strategy decomposition~$\omega$
  of~$U$ for Dis.
  Note that $\big((U, \omega), (T, \tau), (R'',\omega'')\big)$ is 
  a strategy decomposition of~$V$ for Dis.
\end{proof}

\subsection{From strategy decompositions to progress measures} 

\begin{lemma}
\label{lemma:decomposition-to-measure}
  If there is a strategy decomposition of~$V$ for Dis then there is a
  progress measure. 
\end{lemma}

\begin{proof}
  The proof is by induction on the size of the game graph. 
  Let $b$ be the biggest vertex priority in~$V$. 
  We strengthen the inductive hypothesis by requiring that the
  progress measure $(\mu, \varphi)$ whose existence is claimed in
  the lemma is such that
  all sequences in the image of~$\mu$ have the same prefix corresponding 
  to indices~$k$, such that $k > b$. 
  We need to consider two cases based on the parity of~$b$. 

Suppose first that $b$ is even.
  Let $\omega = \big((R, \omega'), (T, \tau), B\big)$ be a
  $b$-decomposition of~$V$ for Dis.
  Since $B \not= \emptyset$, by the inductive hypothesis there is a
  progress measure $(\mu', \varphi')$ on~$R$.
  For every vertex~$v \in T$, define its $\tau$-distance
  to~$B$ to be the largest number of edges on a path starting at~$v$,
  consistent with~$\tau$, and whose only vertex in~$B$ is the last
  one. 
  Let $k$ be the largest such $\tau$-distance, and we
  define $T_i$, $1 \leq i \leq k$, to be the set of vertices in~$T$
  whose $\tau$-distance to~$B$ is~$i$. 

  Let $\seq{m_{d-1}, m_{d-3}, \dots, m_{b+1}}$ be the common prefix of
  all sequences in the image of~$\mu'$. 
  Let $t_1, t_2, \dots, t_k$ be elements of the linearly ordered set
  used in progress measurements, such that for every $r$ that is 
  the component of a sequence in the image of~$\mu'$ corresponding to
  priority~$b-1$, we have 
  $r > t_k > \dots > t_2 > t_1$, and let $t$ be a chosen element of the linearly ordered
  set (it does not matter which one).
  Define the progress labelling $(\mu, \varphi)$ for all
  vertices $v \in V$ as follows:
  \[
  \big(\mu(v), \varphi(v)\big) =
  \begin{cases}
    \big(\mu'(v), \varphi'(v)\big) \hfill \text{if $v \in R$}, \\
    \big(\seq{m_{d-1}, \dots, m_{b+1}, t_i, m_{b-3}, \ldots,
    m_{\ell}}, \infty\big) \quad \; \\
      \hfill \text{if $v \in T_i, 1 \leq i \leq k$}, \\ 
    \big(\seq{m_{d-1}, m_{d-3}, \dots, m_{b+1}}, \infty\big) \hfill 
      \text{if $v \in B$}; 
  \end{cases}
  \]
where $\ell$ is the smallest odd number no smaller than $\pi(v)$ 
  and $m_{b-3} = \ldots = m_{\ell} = t$.

 The progress labelling $(\mu, \varphi)$ as defined above is a
  desired progress measure. 
  It is illustrated as a tree in Figure~\ref{figure:case1}.
  
  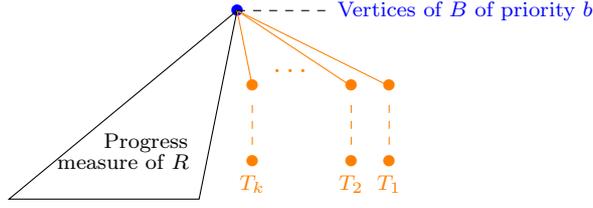
\begin{figure}
  \begin{center}
\begin{tikzpicture}
\node[] at (10,2.5) {\textcolor{blue}{$\bullet$}};
\draw[] (10,2.5) -- (7,0) -- (9.5,0) -- (10,2.5);
\node[] at (8.8,0.75) {\scriptsize Progress}; 
\node[] at (8.5,0.50) {\scriptsize measure of $R$};
\draw[orange] (10,2.5) -- (12,1.5);
\node[] (a) at (12,1.5) {\textcolor{orange}{$\bullet$}};
\draw[orange] (10,2.5) -- (11.5,1.5);
\node[] (b) at (11.5,1.5) {\textcolor{orange}{$\bullet$}};
\node[] at (10.7,1.7) {\textcolor{orange}{$\ldots$}};
\draw[orange] (10,2.5) -- (10.2,1.5);
\node[] (c) at (10.2,1.5) {\textcolor{orange}{$\bullet$}};

\node[] (d) at (13,2.5) {\scriptsize \textcolor{blue}{Vertices of $B$ of priority $b$}};
\draw[dashed] (10,2.5) -- (d);

\node[] (e) at (12,0.5) {\textcolor{orange}{$\bullet$}};
\node[] at (12,0.2) {\scriptsize \textcolor{orange}{$T_1$}};
\draw[dashed, orange] (e) -- (a);
\node[] (f) at (11.5,0.5) {\textcolor{orange}{$\bullet$}};
\node[] at (11.5,0.2) {\scriptsize \textcolor{orange}{$T_2$}};
\draw[dashed, orange] (f) -- (b);
\node[] (g) at (10.2,0.5) {\textcolor{orange}{$\bullet$}};
\node[] at (10.2,0.2) {\scriptsize \textcolor{orange}{$T_k$}};
\draw[dashed,orange] (g) -- (c);
\end{tikzpicture}
\caption{\label{figure:case1}Construction of a progress measure - $b$ even (the common prefix is not pictured).}
\end{center}
\end{figure}
  
  Suppose now that $b$ is odd.
  Let $\omega = \big((U, \omega''), (T, \tau), (R, \omega')\big)$ be a
  $b$-decomposition of~$V$ for Dis.
  Define $\tau$-distances, sets~$T_i$, and elements $t_i$ and $t$
  for $1 \leq i \leq k$, in the analogous way to the ``even~$b$''
  case, replacing set~$B$ by set~$R$.
  By the inductive hypothesis, there is a progress measure
  $(\mu'', \varphi'')$ on~$U$, and let
  $\seq{m_{d-1}, m_{d-3}, \dots, m_{b+2}}$ be the common prefix of
  all sequences in the image of~$\mu''$.
  We define a progress labelling $(\mu, \varphi)$ for all vertices
  in $U \cup T$ as follows:
  \[
  \big(\mu(v), \varphi(v)\big) =
  \begin{cases}
    \big(\mu''(v), \varphi''(v)\big) \hfill \text{if $v \in U$}, \\
    \big(\seq{m_{d-1}, \dots, m_{b+2}, t_i, m_{b-3}, \ldots,
      m_{\ell}}, \infty\big) \quad \; \\
    \hfill \text{if $v \in T_i$, $1 \leq i \leq k$};
  \end{cases}
  \]
  where $\ell$ is the smallest odd number no smaller than $\pi(v)$
  and $m_{b-3} = \ldots = m_{\ell} = t$.

  If $\omega'$ is a $b'$-decomposition of~$R$ for $b' < b$ (case~\ref{item:win-by-inductionDis}), then by
  the inductive hypothesis, there is a progress measure
  $(\mu', \varphi')$ on~$R$. 
  Without loss of generality, assume that all sequences in the images
  of~$\mu'$ and of~$\mu''$ have the common prefix
  $\seq{m_{d-1}, m_{d-3}, \dots, m_{b+2}}$, and that for all~$u$
  and~$r$ that are the components of a sequence in the images
  of~$\mu''$ and~$\mu'$, respectively, corresponding to priority~$b$,
  we have $u > t_k > t_{k-1} > \dots > t_1 > r$. 
  Define the progress labelling $(\mu, \varphi)$ for all vertices
  $v \in R$ in the following way:
  \[
  \big(\mu(v), \varphi(v)\big) = \big(\mu'(v), \varphi'(v)\big).
  \]
  This is illustrated in Figure~\ref{figure:case2}.
  
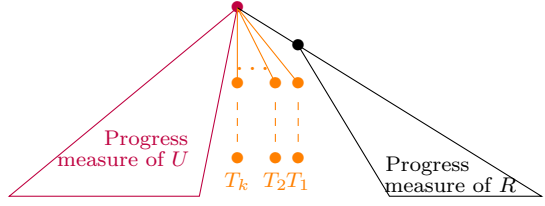
\begin{figure}
\begin{center}
\begin{tikzpicture}
\node[] at (10,2.5) {\textcolor{purple}{$\bullet$}};
\draw[purple] (10,2.5) -- (7,0) -- (9.5,0) -- (10,2.5);
\node[] at (8.8,0.75) {\scriptsize \textcolor{purple}{Progress}}; 
\node[] at (8.5,0.50) {\scriptsize \textcolor{purple}{measure of $U$}};
\draw[orange] (10,2.5) -- (10.8,1.5);
\node[] (a) at (10.8,1.5) {\textcolor{orange}{$\bullet$}};
\draw[orange] (10,2.5) -- (10.5,1.5);
\node[] (b) at (10.5,1.5) {\textcolor{orange}{$\bullet$}};
\node[] at (10.2,1.7) {\textcolor{orange}{$\ldots$}};
\draw[orange] (10,2.5) -- (10,1.5);
\node[] (c) at (10,1.5) {\textcolor{orange}{$\bullet$}};

\node[] (e) at (10.8,0.5) {\textcolor{orange}{$\bullet$}};
\node[] at (10.8,0.2) {\scriptsize \textcolor{orange}{$T_1$}};
\draw[dashed, orange] (e) -- (a);
\node[] (f) at (10.5,0.5) {\textcolor{orange}{$\bullet$}};
\node[] at (10.5,0.2) {\scriptsize \textcolor{orange}{$T_2$}};
\draw[dashed, orange] (f) -- (b);
\node[] (g) at (10,0.5) {\textcolor{orange}{$\bullet$}};
\node[] at (10,0.2) {\scriptsize \textcolor{orange}{$T_k$}};
\draw[dashed, orange] (g) -- (c);

\node[] at (10.8,2) {$\bullet$};
\draw[] (10,2.5) -- (10.8,2);
\draw[] (10.8,2) -- (12,0) -- (14,0) -- (10.8,2);
\node[] at (12.5,0.4) {\scriptsize Progress}; 
\node[] at (12.8,0.15) {\scriptsize measure of $R$};
\end{tikzpicture}
\caption{\label{figure:case2}Construction of a progress measure - $b$ odd - case~\ref{item:win-by-inductionDis}.}
\end{center}
\end{figure}

  If, instead, $\omega'$ is a positional strategy for Dis that
  is mean-payoff winning for him on~$R$ (case~\ref{item:win-by-energyDis}), 
  then by the result of Brim et al.~\cite{BCDGR11}, there is an energy
  progress measure $\widehat{\varphi}$ for Dis on~$R$.
  Let $r'$ be such that $r' < t_1$, and define the progress labelling 
  $(\mu, \varphi)$ for all vertices $v \in R$ in the following
  way: 
  \[
  \big(\mu(v), \varphi(v)\big) =
    \big(\seq{m_{d-1}, m_{d-3}, \dots, m_{b+2}, r'}, \widehat{\varphi}(v)\big). 
  \]
  This is illustrated in Figure~\ref{figure:case3}.
  
\begin{figure}
\begin{center}
\begin{tikzpicture}
\node[] at (10,2.5) {$\bullet$};
\draw[purple] (10,2.5) -- (7,0) -- (9.5,0) -- (10,2.5);
\node[] at (8.8,0.75) {\scriptsize \textcolor{purple}{Progress}}; 
\node[] at (8.5,0.50) {\scriptsize \textcolor{purple}{measure of $U$}};
\draw[orange] (10,2.5) -- (10.8,1.5);
\node[] (a) at (10.8,1.5) {\textcolor{orange}{$\bullet$}};
\draw[orange] (10,2.5) -- (10.5,1.5);
\node[] (b) at (10.5,1.5) {\textcolor{orange}{$\bullet$}};
\node[] at (10.2,1.7) {\textcolor{orange}{$\ldots$}};
\draw[orange] (10,2.5) -- (10,1.5);
\node[] (c) at (10,1.5) {\textcolor{orange}{$\bullet$}};

\node[] (e) at (10.8,0.5) {\textcolor{orange}{$\bullet$}};
\node[]  at (10.8,0.2) {\scriptsize \textcolor{orange}{$T_1$}};
\draw[dashed, orange] (e) -- (a);
\node[] (f) at (10.5,0.5) {\textcolor{orange}{$\bullet$}};
\node[] at (10.5,0.2) {\scriptsize \textcolor{orange}{$T_2$}};
\draw[dashed, orange] (f) -- (b);
\node[] (g) at (10,0.5) {\textcolor{orange}{$\bullet$}};
\node[] at (10,0.2) {\scriptsize \textcolor{orange}{$T_k$}};
\draw[dashed, orange] (g) -- (c);

\node[] (d) at (12,2.5) {\scriptsize Vertices of $R$};
\draw[dashed] (10,2.5) -- (d);
\end{tikzpicture}
\caption{\label{figure:case3}Construction of a progress measure - case~\ref{item:win-by-energyDis}.}
\end{center}
\end{figure}
  
The progress labelling $(\mu, \varphi)$ as defined above is a
  desired progress measure.
\end{proof}

\section{Computing progress measures by lifting}
\label{section:lifting}
In this section, we give a so-called lifting algorithm which identifies the winning sets for Dis and for Con by computing a progress measure on the winning set for Dis.

By the \emph{tree} of a progress labelling $(\mu, \varphi)$, we mean the ordered tree whose nodes are all prefixes of all sequences $\mu(v)$ as $v$ ranges over the vertices of the game graph, and such that every vertex $v$ labels the node $\mu(v)$ of the tree.  Let us say that progress labellings $(\mu, \varphi)$ and $(\mu', \varphi')$ are \emph{isomorphic} if and only if their (partially labelled ordered) trees are isomorphic and $\varphi = \varphi'$. 

We shall work with the following ordering on finite binary strings:
\[0 s < \varepsilon, \quad
  \varepsilon < 1 s, \quad
  b s < b s' \text{ if and only if } s < s',\]
where $\varepsilon$ denotes the empty string, $b$ ranges over binary digits, and $s, s'$ range over binary strings.

Recall that $n$ is the number of vertices, and $d$ (assumed even) is the number of priorities.  

Let $S_{n, d}$ be all sequences $\seq{m_{d - 1}, m_{d - 3}, \ldots, m_{\ell}}$ of binary strings such that:
\begin{itemize}
\item
$\ell$ is odd and $1 \leq \ell \leq d + 1$;
\item
$\sum_{i = \ell}^{d - 1} |m_i| \,\leq\, \lceil \lg n \rceil$;
\end{itemize}
and let us call a progress measurement, labelling or measure \emph{succinct} if and only if all the sequences $\seq{m_{d - 1}, m_{d - 3}, \ldots, m_{\ell}}$ involved are members of~$S_{n, d}$.

\begin{lemma}
\label{lemma:succinct-isomorphic-progress-labelling}
For every progress labelling, there exists a succinct isomorphic one.
\end{lemma}

\begin{proof}
This is an immediate consequence of \cite[Lemma~1]{JL17}, since for every progress labelling, its tree is of height at most $d / 2$ and has at most $n$ leaves.
\end{proof}

\begin{corollary}
\label{corollary:restricted-to-succinct}
Lemmas \ref{lemma:measure-to-decomposition} and \ref{lemma:decomposition-to-measure} hold when restricted to succinct progress measures.
\end{corollary}

We now order progress measurements lexicographically:
\begin{multline*}
\big(\seq{m_{d - 1}, m_{d - 3}, \ldots, m_{\ell}}, e\big) < \big(\seq{m'_{d - 1}, m'_{d - 3}, \ldots, m'_{\ell'}}, e'\big) \\ 
\text{ if and only if} \\
\text{either } \seq{m_{d - 1}, m_{d - 3}, \ldots, m_{\ell}} < \seq{m'_{d - 1}, m'_{d - 3}, \ldots, m'_{\ell'}}, \\
\text{or } \seq{m_{d - 1}, m_{d - 3}, \ldots, m_{\ell}} = \seq{m'_{d - 1}, m'_{d - 3}, \ldots, m'_{\ell'}} \text{ and } e < e'
\end{multline*}
and we extend them by a new greatest progress measurement $(\top, \infty)$.
We then revise the set of progress labellings to allow the extended progress measurements,
and we (partially) order it pointwise:
\begin{multline*}
(\mu, \varphi) \leq (\mu', \varphi')
\text{ if and only if,} \\
\text{ for all $v \in V$,} 
\big(\mu(v), \varphi(v)\big) \leq \big(\mu'(v), \varphi(v')\big). 
\end{multline*}
We also revise the definition of a progress measure by stipulating that an edge $(v, u)$ which involves the progress measurement $(\top, \infty)$ is progressive if and only if the progress measurement of $v$ is $(\top, \infty)$.

For any succinct progress labelling $(\mu, \varphi)$ and edge $(v, u)$, we set $\mathrm{lift}(\mu, \varphi, v, u)$ to be the minimum succinct progress measurement $\big(\seq{m_{d - 1}, m_{d - 3}, \ldots, m_{\ell}}, e\big)$ which is at least $\big(\mu(v), \varphi(v)\big)$ and such that $(v, u)$ is progressive in the updated succinct progress labelling
\[\Big(\mu\big[v \mapsto \seq{m_{d - 1}, m_{d - 3}, \ldots, m_{\ell}}\big], \varphi[v \mapsto e]\Big)\;.\]
For any vertex $v$, we define an operator $\mathrm{Lift}_v$ on succinct progress labellings as follows:
\[\mathrm{Lift}_v(\mu, \varphi)(w) =
\begin{cases}
\big(\mu(w), \varphi(w)\big) & \text{if } w \neq v, \\
\min_{(v, u) \in E} \mathrm{lift}(\mu, \varphi, v, u) & \text{if Dis owns } w = v, \\
\max_{(v, u) \in E} \mathrm{lift}(\mu, \varphi, v, u) & \text{if Con owns } w = v.
\end{cases}\]

\begin{theorem}[Correctness of lifting algorithm]
\ 
\begin{enumerate}
\item 
The set of all succinct progress labellings ordered pointwise is a complete lattice.
\item
Each operator $\mathrm{Lift}_v$ is inflationary and monotone.
\item
From every succinct progress labelling $(\mu, \varphi)$, 
every sequence of applications of operators $\mathrm{Lift}_v$
eventually reaches the least simultaneous fixed point of all
$\mathrm{Lift}_v$ that is greater than or equal to~$(\mu, \varphi)$.
\item
A succinct progress labelling $(\mu, \varphi)$ is a simultaneous fixed point of all operators $\mathrm{Lift}_v$
if and only if it is a succinct progress measure.
\item
If $(\mu^*, \varphi^*)$ is the least succinct progress measure,
then $\{v \,:\, \big(\mu^*(v), \varphi^*(v)\big) \neq (\top, \infty)\}$ is the set of winning positions for Dis.
\end{enumerate}
\end{theorem}

\begin{proof}
\begin{enumerate}
\item 
  The partial order of all succinct progress labellings is 
  the pointwise product of $n$ copies of 
  the finite linear order of all succinct progress measurements.

\item
  We have inflation, i.e.
  $$\mathrm{Lift}_v(\mu, \varphi)(w) \geq \big(\mu(w), \varphi(w)\big)$$ 
  by  
  the definitions of $\mathrm{Lift}_v(\mu, \varphi)(w)$ and
  $\mathrm{lift}(\mu, \varphi, v, u)$.
  
  For monotonicity, supposing $(\mu, \varphi) \leq (\mu', \varphi')$, 
  it suffices to show that, for every edge $(v, u)$, we have
  $\mathrm{lift}(\mu, \varphi, v, u) \leq \mathrm{lift}(\mu', \varphi', v, u)$, which is in turn implied by the straightforward observation that, whenever an edge is progressive with respect to a progress labelling, it remains progressive after any lessening of the progress measurement of its target vertex.

\item
  This holds for any family of 
  inflationary monotone operators on a finite complete lattice.
  Consider any such maximal sequence from $(\mu, \varphi)$.
  It is an upward chain from $(\mu, \varphi)$ to some $(\mu^*, \varphi^*)$ 
  which is a simultaneous fixed point of all the operators.
  For any $(\mu', \varphi') \geq (\mu, \varphi)$ which is also a simultaneous fixed point,
  a simple induction confirms that $(\mu^*, \varphi^*) \leq (\mu', \varphi')$.

\item
  Here we have a rewording of the definition of a succinct progress measure.

\item
  Let $W = \{v \,:\, \big(\mu^*(v), \varphi^*(v)\big) \neq (\top, \infty)\}$.
  The set of winning positions for Dis is contained in~$W$
  by Lemma~\ref{lemma:existence-strategy-decomposition},
  Lemma~\ref{lemma:decomposition-to-measure} and
  Corollary~\ref{corollary:restricted-to-succinct}, because
  $(\mu^*, \varphi^*)$ is the least succinct progress measure. 
  
  Since $(\mu^*, \varphi^*)$ is a progress measure, 
  we have that, for every progressive edge $(v, u)$, 
  if $\big(\mu^*(v), \varphi^*(v)\big) \neq (\top, \infty)$ 
  then $\big(\mu^*(u), \varphi^*(u)\big) \neq (\top, \infty)$.
  In order to show that Dis has a winning strategy from every vertex in~$W$, 
  it remains to apply Lemmas \ref{lemma:measure-to-decomposition} and 
  \ref{lemma:strategy-for-Dis} to the subgame consisting of the vertices 
  in~$W$. 
\end{enumerate}
\end{proof}

\begin{table}
\begin{center}
  \fbox{\parbox{0.9\columnwidth}{
   \begin{enumerate}
     \item Initialise $(\mu, \varphi)$ to the least succinct progress labelling 
    $$(v \mapsto \seq{}, v \mapsto 0)$$

    \item While $\mathrm{Lift}_v(\mu, \varphi) \neq (\mu, \varphi)$ for some $v$,
    update $(\mu, \varphi)$ to become $\mathrm{Lift}_v(\mu, \varphi)$.

    \item Return the set $W_{\mathrm{Dis}} = \{v \,:\, \big(\mu(v), \varphi(v)\big) \neq (\top, \infty)\}$ of winning positions for Dis.
    \end{enumerate}
  }}
  \caption{The lifting algorithm.}
  \label{table:algorithm}
  \end{center}
\end{table}

\begin{lemma}[Jurdzi\'nski and Lazi\'c~\cite{JL17}]
  \label{lemma:size-of-Snd}
  Depending on the asymptotic growth of~$d$ as a function of~$n$, 
  the size of the set $S_{n, d}$ is as follows:
  \begin{enumerate}
  \item
    $O\left(n^{1+o(1)}\right)$ if $d = o(\log n)$;

  \item
    \label{enumerate:d-delta-log-n}
    $\Theta\left(n^{\lg(\delta+1) + \lg(e_\delta) + 1} \middle/
      \sqrt{\log n}\right)$
    if $d/2 = \lceil \delta \lg n \rceil$,
    for some positive constant~$\delta$, and 
    where $e_{\delta} = (1 + 1/\delta)^\delta$; 

  \item
    \label{enumerate:d-omega-log-n}
    $O\left(d n^{\lg (d/{\lg n}) + \lg e + o(1)}\right)$ if $d = \omega(\log n)$. 
  \end{enumerate}  
\end{lemma}

\begin{theorem}[Complexity of lifting algorithm]
  Depending on the asymptotic growth of~$d$ as a function of~$n$, 
  the running time of the algorithm is as follows:
  \begin{enumerate}
  \item
    \label{enumerate:d-o-lg-n}
    $O\left(m n^{2+o(1)} C\right)$ if $d = o(\log n)$;

  \item
    \label{enumerate:d-less-lg-n}
    $O\left(m n^{\lg(\delta+1) + \lg(e_\delta) + 2} C \cdot \log d \cdot \sqrt{\log n}\right)$
    if $d \leq 2 \lceil \delta \lg n \rceil$, for some positive
    constant~$\delta$;

  \item
    \label{enumerate:d-omega-lg-n}
    $O\left(dm n^{\lg(d/{\lg n}) + 2.45} C\right)$ if $d = \omega(\log \eta)$. 
  \end{enumerate}
  The algorithm works in space $O(n \cdot \log n \cdot \log d)$.
\end{theorem}

\begin{proof}
The work space requirement is dominated by the number of bits needed to store a single succinct progress labelling, which is at most $n (\lceil \lg n \rceil \lceil \lg d \rceil + \lceil \lg (n C) \rceil)$.

Since bounded-depth successors of elements of $S_{n, d}$ are computable in time $O(\log n \cdot \log d)$ (cf.\ the proof of \cite[Theorem~7]{JL17}, the $\mathrm{Lift}_v$ operators can be implemented to work in time $O(\mathrm{deg}(v) \cdot (\log n \cdot \log d + \log C))$.
It then follows, observing that the algorithm lifts each vertex at most 
$|S_{n, d}| (n C + 1)$ times, that its running time is bounded by
\begin{multline*}
  O\left(\sum_{v \in V} 
    \mathrm{deg}(v) \cdot (\log n \cdot \log d + \log C)
    |S_{n, d}| (n C + 1)\right) =
  \\ 
  O\left(m n C (\log n \cdot \log d + \log C) |S_{n, d}|\right)\;.
\end{multline*}
From there, the various stated bounds are obtained by applying 
Lemma~\ref{lemma:size-of-Snd},
and by suppressing some of the multiplicative factors that are 
logarithmic in the bit-size of the input.
Suppressing the $\log C$ factor is justified by using the
unit-cost RAM model, which is the industry standard in algorithm
analysis.
The reasons for suppressing the $\log n$ and $\log d$ factors are
more varied:
in case~\ref{enumerate:d-o-lg-n}, they are absorbed by the $o(1)$ 
term in the exponent of~$n$, and 
in case~\ref{enumerate:d-omega-lg-n}, they are absorbed in the
$2.45$ term in the exponent of~$n$, because $\lg e < 1.4427$. 
\end{proof}

\section{From winning sets to strategy decompositions for Con}
\label{section:winsets-to-decompositions}
The pseudo-quasi-polynomial lifting algorithm computes the least progress
measure and hence, by Lemmas~\ref{lemma:measure-to-decomposition}
and~\ref{lemma:strategy-for-Dis}, it can be easily adapted to
synthesize a winning strategy for Dis from all vertices in her winning
set.
In this section we tackle the problem of strategy synthesis for Con.
By (the proof of) Lemma~\ref{lemma:strategy-for-Con}, in order to
synthesize a winning strategy for Con, it suffices to compute a
strategy decomposition for him.
We argue that this can also be achieved in pseudo-quasi-polynomial time. 

\begin{theorem}[Complexity of computing strategy decompositions]
\label{theorem:pseudo-quasi-decompositions-for-Con}
  There is a pseudo-quasi-polynomial algorithm that computes strategy
  decompositions for both players on their winning sets. 
\end{theorem}

In order to establish that strategy decompositions for Con can be 
computed in pseudo-quasi-polynomial time, it suffices to prove the 
following lemma, because the polynomial-time oracle algorithm becomes
a pseudo-quasi-polynomial algorithm, once the oracle for computing
winning strategies in mean-payoff games is replaced by a 
pseudo-polynomial algorithm~\cite{ZP96,BCDGR11,CR17}, and the oracle
for computing the winning sets in mean-payoff parity games is 
replaced by the pseudo-quasi-polynomial procedure from 
Section~\ref{section:lifting}. 

\begin{lemma}
  \label{lemma:computing-strategy-decompositions-for-Con}
  There is a polynomial-time algorithm, with oracles for computing 
  winning strategies in mean-payoff games and for computing 
  winning sets in mean-payoff parity games, that computes a strategy
  decomposition for Con of his winning set. 
\end{lemma}

\begin{proof}
  Without loss of generality, we may assume that Con has a
  winning strategy from every vertex in~$V$, since a single call to
  the oracle allows us to reduce~$V$ to the subgame corresponding to
  the winning set for Con.

  Below, we describe a recursive procedure for computing a strategy
  decomposition for Con of the set of all vertices, that has a similar 
  structure to the inductive proof of
  Lemma~\ref{lemma:existence-strategy-decomposition}.
  In parallel with the description of the recursive procedure, we
  elaborate an inductive proof that it does indeed compute a strategy
  decomposition for Con on~$V$.
  
  Note that our procedure avoids incurring the penalty of adding to
  its running time a factor that is exponential in the number of
  distinct vertex priorities, by repeatedly using the oracle for 
  computing the winning sets in appropriately chosen subgames.
  We give a detailed analysis of the worst-case running time at the
  end of this proof. 

  Let $B$ be the set of vertices of the highest priority~$b$;
  let $T$ be the set of vertices (not including vertices in~$B$) from 
  which Dis has a strategy to reach a vertex in~$B$;
  let $\tau$ be a corresponding positional reachability strategy;
  and let $R = V \setminus (B \cup T)$.
  We consider two cases, depending on the parity of~$b$.

  \paragraph{Even $b$.} We can assume that $R\neq \emptyset$, otherwise Dis would win the game, which contradicts the assumption that
  Con has a winning strategy from every vertex.
  Call the oracle to obtain the partition~$R_{\mathrm{Con}}$
  and~$R_{\mathrm{Dis}}$ of~$R$, the winning sets for Con and for Dis,
  respectively, in the subgame~$R$.
  We argue that $R_{\mathrm{Con}} \not= \emptyset$.
  Otherwise, by Lemma~\ref{lemma:existence-strategy-decomposition},
  there is a strategy decomposition~$\omega$ of~$R$ for Dis, and 
  hence $\big((R, \omega), (T, \tau), B\big)$ is a strategy
  decomposition of~$V$ for Dis, which, by
  Lemma~\ref{lemma:strategy-for-Dis}, contradicts the assumption that
  Con has a winning strategy from every vertex. 

  Let $T'$ be the set of vertices
  (not including vertices in~$R_{\mathrm{Con}}$)
  from which Con has a strategy to reach a vertex
  in~$R_{\mathrm{Con}}$, and let $\tau'$ be a corresponding positional
  reachability strategy, and let
  $U = V \setminus (R_{\mathrm{Con}} \cup T')$.
  By the inductive hypothesis, a recursive call of our procedure
  on $R_{\mathrm{Con}}$ will produce a strategy
  decomposition~$\omega'$ of~$R_{\mathrm{Con}}$ for Con, and another
  recursive call of the procedure on~$U$ will produce a strategy
  decomposition~$\omega''$ of~$U$ for Con.
  We claim that
  $\big((U, \omega''), (T', \tau'), (R_{\mathrm{Con}}, \omega')\big)$
  is a strategy decomposition of~$V$ for Con.

  \paragraph{Odd $b$.}
  Call the oracle for computing positional winning strategies in 
  mean-payoff games to obtain a positional strategy~$\lambda$ for Con
  that is mean-payoff winning for him on~$V$;
  such a strategy exists because Con has a mean-payoff parity winning 
  strategy from every vertex.
  Since $R$ is a trap for Con, it must be the case that Con has a
  winning strategy from every vertex in the subgame~$R$. 
  By the inductive hypothesis, a recursive call of our procedure
  on~$R$ will produce a strategy decomposition~$\omega'$ of~$R$ for
  Con.
  We claim that $\big((R, \omega'), (T, \tau), B, \lambda\big)$ is a
  strategy decomposition of~$V$ for Con.
  
  \paragraph{} 
  It remains to argue that the recursive procedure described above
  works in polynomial time in the worst case.
  Observe that in both cases considered above, a call of the procedure
  on a game results in two or one recursive calls, respectively. 
  In both cases, the recursive calls are applied to subgames with
  strictly fewer vertices, and---crucially for the complexity 
  analysis---in the former case, the two recursive calls are applied 
  to subgames on disjoint sets of vertices.
  Additional work (other than recursive calls and oracle calls) in both 
  cases can be bounded by $O(m)$, since the time needed is dominated by 
  the worst case bound on the computation of reachability strategies. 
  Overall, the running time function $T(n)$ of the recursive
  procedure, where $n$ is the number of vertices in the input game
  graph, satisfies the following recurrence:
  \[
  T(n) \leq T(n') + T(n'') + O(m), \qquad \text{where $n' + n'' < n$},
  \]
  and hence $T(n) = O(nm)$. 
\end{proof}

\section{Conclusion}
\label{section:conclusion}
Our main result is the first pseudo-quasi-polynomial algorithm 
for computing the values of mean-payoff parity games, and hence also
for deciding the winner in energy parity games and in parity games
with weights. 
The main technical tools that we introduce to achieve the main
result are strategy decompositions and progress measures for the 
threshold version of mean-payoff games. 
We believe that our techniques can be adapted to also produce 
optimal strategies for both players 
(i.e., the strategies that secure the value that we show how 
to compute).
Another direction for future work is improving the complexity
of solving stochastic mean-payoff parity games~\cite{CDGO14}. 

\section*{Acknowledgements}

This research has been supported by the EPSRC grant 
EP/P020992/1 (Solving Parity Games in Theory and Practice).

\bibliographystyle{plain}
\bibliography{main}

\begin{thebibliography}{10}

\bibitem{BMOU11}
P.~Bouyer, N.~Markey, J.~Olschewski, and M.~Ummels.
\newblock Measuring permissiveness in parity games: {M}ean-payoff parity games
  revisited.
\newblock In {\em ATVA}, pages 135--149, 2011.

\bibitem{BCDGR11}
L.~Brim, J.~Chaloupka, L.~Doyen, R.~Gentilini, and J.-F. Raskin.
\newblock Faster algorithms for mean-payoff games.
\newblock {\em Form. Methods Syst. Des.}, 38(2):97--118, 2011.

\bibitem{CJKLS17}
C.~S. Calude, S.~Jain, B.~Khoussainov, W.~Li, and F.~Stephan.
\newblock Deciding parity games in quasipolynomial time.
\newblock In {\em STOC}, pages 252--263, 2017.

\bibitem{CD12}
K.~Chatterjee and L.~Doyen.
\newblock Energy parity games.
\newblock {\em Theoretical Computer Science}, 458:49--60, 2012.

\bibitem{CDGO14}
K.~Chatterjee, L.~Doyen, H.~Gimbert, and Y.~Oualhadj.
\newblock Perfect-information stochastic mean-payoff parity games.
\newblock In {\em FOSSACS}, pages 210--225, 2014.

\bibitem{CHS17}
K.~Chatterjee, M.~Henzinger, and A.~Svozil.
\newblock Faster algorithms for mean-payoff parity games.
\newblock In {\em MFCS}, pages 39:1--39:17, 2017.

\bibitem{CHJ05}
K.~Chatterjee, T.~A. Henzinger, and M.~Jurdzi\'nski.
\newblock Mean-payoff parity games.
\newblock In {\em LICS}, pages 178--187, 2005.

\bibitem{CR17}
C.~Comin and R.~Rizzi.
\newblock Improved pseudo-polynomial bound for the value problem and optimal
  strategy synthesis in mean payoff games.
\newblock {\em Algorithmica}, 77(4):995--1021, 2017.

\bibitem{Con92}
A.~Condon.
\newblock The complexity of stochastic games.
\newblock {\em Information and Computation}, 96(2):203--224, 1992.

\bibitem{EM79}
A.~Ehrenfeucht and J.~Mycielski.
\newblock Positional strategies for mean payoff games.
\newblock {\em Journal of Game Theory}, 8(2):109--113, 1979.

\bibitem{EJ91}
E.~A. Emerson and C.~Jutla.
\newblock Tree automata, mu-calculus and determinacy.
\newblock In {\em FOCS}, pages 368--377, 1991.

\bibitem{EJS01}
E.~A. Emerson, C.~Jutla, and A.~P. Sistla.
\newblock On model-checking for fragments of $\mu$-calculus.
\newblock {\em Theoretical Computer Science}, 258(1--2):491--522, 2001.

\bibitem{FJSSW17}
J.~Fearnley, S.~Jain, S.~Schewe, F.~Stephan, and D.~Wojtczak.
\newblock An ordered approach to solving parity games in quasi polynomial time
  and quasi linear space.
\newblock In {\em SPIN}, pages 112--121, 2017.

\bibitem{FZ14}
N.~Fijalkow and M.~Zimmermann.
\newblock Parity and {S}treett games with costs.
\newblock {\em Logical Methods in Computer Science}, 10(1:14):1--29, 2014.

\bibitem{GH82}
Y.~Gurevich and L.~Harrington.
\newblock Trees, automata, and games.
\newblock In {\em STOC}, pages 60--65, 1982.

\bibitem{Joh07}
D.~S. Johnson.
\newblock The np-completeness column: {F}inding needles in haystacks.
\newblock {\em ACM Transactions on Algorithms}, 3(2), 2007.

\bibitem{Jur00}
M.~Jurdzi\'nski.
\newblock Small progress measures for solving parity games.
\newblock In {\em STACS}, pages 290--301, 2000.

\bibitem{JL17}
M.~Jurdzi\'nski and R.~Lazi\'c.
\newblock Succinct progress measures for solving parity games.
\newblock In {\em LICS}, pages 1--9, 2017.

\bibitem{KK95}
N.~Klarlund and D.~Kozen.
\newblock Rabin measures.
\newblock {\em Chicago Journal of Theoretical Computer Science}, 1995.
\newblock Article 3.

\bibitem{McN93}
R.~McNaughton.
\newblock Infinite games played on finite graphs.
\newblock {\em Annals of Pure and Applied Logic}, 65(2):149--184, 1993.

\bibitem{SWZ18}
S.~Schewe, A.~Weinert, and M.~Ziemmermann.
\newblock Parity games with weights.
\newblock arXiv:1804.06168, 2018.

\bibitem{Tho95}
W.~Thomas.
\newblock On the synthesis of strategies in infinite games.
\newblock In {\em STACS}, pages 1--13, 1995.

\bibitem{Zie98}
W.~Zielonka.
\newblock Infinite games on finitely coloured graphs with applications to
  automata on infinite trees.
\newblock {\em Theoretical Computer Science}, 200:135--183, 1998.

\bibitem{ZP96}
U.~Zwick and M.~Paterson.
\newblock The complexity of mean-payoff games on graphs.
\newblock {\em Theoretical Computer Science}, 158:343--359, 1996.

\end{thebibliography}

\end{document}